\documentclass[12pt]{article}

\usepackage[english]{babel}	    

\usepackage{amsmath, amsfonts, amssymb, amsthm, mathtools, bm, bbm} 
\usepackage{dsfont}

\usepackage[usenames, dvipsnames, svgnames, table, rgb]{xcolor}
\usepackage{colortbl}
\usepackage{chngcntr}
\usepackage{apptools}

\definecolor{BlueBottle}{RGB}{0, 130, 255}
\usepackage{comment}

\numberwithin{equation}{section}

\usepackage{graphicx}                                   
\graphicspath{{images/}{images2/}}               
\usepackage{wrapfig}                                      
\usepackage{subcaption}

\usepackage{array, tabularx, tabulary, booktabs}  
\usepackage{longtable}                                      
\usepackage{multirow}     
\usepackage{makecell}                                

\theoremstyle{plain}                         

\newtheorem{assumption}{Assumption}[section]
\newtheorem{proposition}{Proposition}
\newtheorem{lemma}{Lemma}

\theoremstyle{definition}                  

\newtheorem{remark}{Remark}

\theoremstyle{remark}                      

\AtAppendix{\counterwithin{lemma}{section}}
\AtAppendix{\counterwithin{theorem}{section}}
\AtAppendix{\counterwithin{equation}{section}}

\newcommand{\E}{\mathbb{E}}
\newcommand{\Var}{\mathbb{V}ar}
\newcommand{\Cov}{\mathbb{C}ov}

\renewcommand{\qed}{\hfill$\blacksquare$}

\makeatletter
\renewcommand*\env@matrix[1][\arraystretch]{%
	\edef\arraystretch{#1}%
	\hskip -\arraycolsep
	\let\@ifnextchar\new@ifnextchar
	\array{*\c@MaxMatrixCols c}}
\makeatother

\usepackage{geometry}                  	   
\usepackage{setspace}   
\usepackage{hyperref}
\hypersetup{			                   	   
    colorlinks=true,                 	      
    linkcolor=BlueBottle,                         
    citecolor=BlueBottle,                          
   urlcolor=BlueBottle                               
}

\usepackage[round, longnamesfirst]{natbib}   
\usepackage{har2nat}

\usepackage{multicol}                      

\usepackage{pgf, tikz}                           
\usepackage{pgfplots}
\usepackage{pgfplotstable}
\usetikzlibrary{arrows}
\usetikzlibrary[patterns]
\usetikzlibrary{arrows.meta}
\usepackage{pdfpages}

\usepackage{caption}
\usepackage{subcaption}
\usepackage{cleveref}
\usepackage{graphicx}

\usepackage{lscape}

\usetikzlibrary{arrows.meta}

\title{On the Lower Confidence Band for the Optimal Welfare in Policy Learning }
\author{ Kirill Ponomarev\thanks{Department of Economics, University of Chicago. Email: kponomarev@uchicago.edu} \and 	Vira Semenova\thanks{Department of Economics, University of California, Berkeley. Email: vsemenova@berkeley.edu. \newline
First version: October 2024, \href{https://arxiv.org/abs/2410.07443}{https://arxiv.org/abs/2410.07443.}  We are grateful to Alberto Abadie, Timothy Armstrong, Isaiah Andrews, Jiafeng  Chen, Victor Chernozhukov, Timothy Christensen, Bryan Graham, Michael Jansson, Patrick Kline, Sokbae Lee, Elena Manresa, Demian Pouzo, Andres Santos, Azeem Shaikh, Liyang  Sun, Vasilis Syrgkanis, Max Tabord-Meehan, Tiemen Woutersen, and participants of 2024 California Econometrics Conference for valuable comments. Jinglin Yang provided superb research assistance. All errors are ours.
}  }

\begin{document}
\maketitle

\begin{abstract}

We study inference on the optimal welfare in a policy learning problem and propose reporting a lower confidence band (LCB). A natural approach to constructing an LCB is to invert a one-sided $t$-test based on an efficient estimator for the optimal welfare. However, we show that for an empirically relevant class of DGPs, such an LCB can be first-order dominated by an LCB based on a welfare estimate for a suitable suboptimal treatment policy. We show that such first-order dominance is possible if and only if the optimal treatment policy is not ``well-separated'' from the rest, in the sense of the commonly imposed margin condition. When this condition fails, standard debiased inference methods are not applicable. We show that uniformly valid and easy-to-compute LCBs can be constructed analytically by inverting moment-inequality tests with the maximum and quasi-likelihood-ratio test statistics. As an empirical illustration, we revisit the National JTPA study and find that the proposed LCBs achieve reliable coverage and competitive length. 
  \end{abstract}
  \maketitle
\noindent\texttt{Keywords:}  Statistical decision theory, policy learning, optimal welfare, lower confidence band, partial identification, sensitivity analysis, cross-fitting,  uniformity \\
\noindent\texttt{JEL Numbers:} C14, C31, C54

\newpage

\section{Introduction}



Treatment assignment problems are ubiquitous in economics, including governments providing subsidies to disadvantaged households, firms offering job training opportunities to their employees, colleges allocating scholarships to students, and online retailers offering discounts to customers. In such settings, a decision-maker (DM) aims to design a treatment rule that determines who should --- and who should not --- be treated, based on observable individual characteristics, to maximize welfare \citep{Manski2004}. Since developing good treatment rules may be costly and time-consuming, the DM might want to quantify the potential welfare gains. To this end, the DM may conduct a preliminary experiment and test a hypothesis that the optimal welfare (or welfare gain) exceeds a certain threshold.

Conducting inference for the optimal welfare (and welfare gain) is a challenging task. From a practical perspective, it may require solving complicated non-convex optimization problems, estimating functions of high-dimensional inputs non-parametrically, and dealing with noisy welfare estimates due to suboptimal experiment design. Theoretically, a major complication is the potential non-uniqueness of the optimal policy, which makes standard debiased inference methods inapplicable \citep{HiranoPorter2012, LuedtkeLaan}.


In this paper, we show that good estimators and tight lower confidence bands (LCBs) for the optimal welfare (and welfare gain) can be obtained by leveraging suboptimal policies. Our first contribution is to demonstrate a possible trade-off between the welfare level and the precision with which it can be estimated in finite samples. For empirically relevant data-generating processes (DGPs), we provide an example of a slightly suboptimal policy, whose welfare can be estimated substantially more precisely than the optimal one. As a result, an LCB targeting such suboptimal welfare can be first-order tighter --- at the $N^{-1/2}$ scale for sample size $N$ --- than the LCB targeting the optimal welfare directly. Additionally, such suboptimal policy yields a better estimator of the optimal welfare in terms of mean-squared error, for all $N$ large enough. In particular, this example shows that incorporating asymptotically redundant information can yield first-order improvements for estimators and inference procedures in finite samples.

Our second contribution is to characterize the class of DGPs for which the first-order trade-off between welfare and precision is possible. Intuitively, if the optimal policy is ``well-separated'' from the rest, the precision gain of any suboptimal policy cannot compensate for the welfare loss. We formalize this intuition using a local asymptotic approximation around a DGP at which ``separation'' fails, and derive minimax rates for the gap between the two LCBs. As a result, we show that the first-order trade-off is possible if and only if the margin condition of \citet{MammenTsybakov} and \citet{Tsybakov} fails to hold uniformly over the relevant DGPs. In such settings, standard debiased inference procedures may be invalid, so alternative inference methods are needed.


To this end, we propose LCBs that address the aforementioned welfare-precision tradeoff and remain valid regardless of the margin condition. The idea is to construct a (possibly large but) finite subclass of test policies, based on economic intuition, within which a ``good'' suboptimal policy may be found. Each of these policies provides a lower bound on the optimal welfare, yielding a collection of moment inequalities that can be tested using existing methods \citep{andrews2010inference, CLR, romano2014practical, canay2017practical}. The existing tests combine self-normalization \citep[precision correction in][]{CLR} with moment selection, leading to tight LCBs that remain valid under relatively weak conditions. For the problem at hand, the tests can often be inverted analytically, so the LCBs are easy to compute in practice. 

To illustrate our theoretical results, we revisit the U.S. National Job Training Partnership Act (JTPA) experiment \citet{Bloom1997}. 
The experiment randomly assigned individuals with distinct education levels and baseline earnings to a job training program and recorded their post-treatment salary. 
For most education years --- apart from graduation thresholds --- the respective conditional average treatment effect is statistically insignificant, 
indicating a violation of the margin condition. Standard procedures that either ignore education or use a holdout sample to estimate the first-best policy suffer from substantial power loss. 
We consider several classes of test policies based only on education and construct the corresponding LCBs by inverting moment-inequality tests as described above. In line with the theoretical predictions, the LCBs are substantially shorter than the available alternatives. 


\medskip 

\noindent 
\textbf{Related Literature} \,
This paper contributes to a large cross-disciplinary literature on optimal treatment choice, following \citet{Manski2004}. In econometrics, contributions range from early program-evaluation and partial-identification approaches to modern policy learning \citep{DEHEJIA2005141,HiranoPorter2009,Stoye,Chamberlain2011,BhatacharyaDupas2012,Tetenov2012,Rai2019,KitagawaTetenov,MbakopTabord,  AtheyWager2, Sun,SasakiUra,KitagawaLeeChen2022,Yata2021,ArmstrongShen2023,chernozhukov2025polece,Moon:25}. In statistics, optimal treatment regimes are commonly learned via Q-learning and A-learning \citep{QianMurphy,Murphy2003,Robins2004,ShiEtAl2018}. This literature focuses primarily on obtaining treatment rules that perform well in terms of expected regret.

In this paper, we consider a complementary problem of inference on the optimal welfare, also studied in \citet{LuedtkeLaan}.\footnote{This problem is distinct from the ``inference on winners'' considered in \citet{andrews2024inference}, \citet{andrewschen2025}, and \citet{chernozhukov2025polece}, and the proposed LCBs are generally not valid in those settings. }  In the absence of ties among the best policies, the authors showed that the optimal welfare is a regular parameter and derived the semiparametric efficiency bound for it. The bound turns out to be the same as if the best policy was known \textit{ex ante}. When ties are present, the optimal welfare is no longer regular \citep{HiranoPorter2012}, but in view of the above, an oracle efficient estimator based on one of the optimal policies still provides a natural benchmark for our analysis. We complement the results of \citet{LuedtkeLaan} by studying MSEs of the estimators and expected length of the associated LCBs in finite samples, formalizing the necessity of the margin condition for one-sided inference, and proposing simple robust inference procedures.  The proposed procedures provide alternatives to the approaches based on smoothing, as in   \citet{chen2023inference}, \citet{levis2023covariateassisted} and \citet{whitehouse2025}, or entropic regularization, as in \citet{benmichael2025}. They also relate to a broader literature on robust policy learning, including decisions under ambiguity \citep{benmichael2021safepolicy,cui2024policy} and concerns about external validity \cite{adjaho2023externally}. Although we focus on the utilitarian (linear) formulation of welfare throughout, the proposed approach also applies in non-linear settings, such as inequality-sensitive welfare studied in \citet{kasy2016partial, kitagawa2021equality, terschuur2025locally}, among others.

This paper also contributes to the literature on inference for partially identified parameters. We show that in finite samples, inference based on sharp bounds may be less precise than inference based on loose bounds, giving rise to a first-order trade-off between sharpness and precision. We argue that existing inference methods are able to address this trade-off by combining self-normalization (precision-correction) and moment selection, while retaining uniform validity \citep[][]{andrews2010inference,CLR,romano2014practical,canay2017practical,bai2022two}.\footnote{A related question of inference with over-identifying inequality constraints is studied, e.g., in \citet{cox2024simple} and \citet{ketz2025short}. Our setting is different in that the target parameter may not be asymptotically Gaussian even when the constraints are not binding.}

The rest of the paper is organized as follows. Section \ref{sec:notation} introduces the policy learning problem and motivates our target parameters. Section \ref{sec:example} gives a sequence of DGPs exhibiting the first-order dominance. Section \ref{sec:mainr2} discusses the role of the margin assumption. Section \ref{sec:moment} proposes robust inference procedures. Section \ref{sec:jc} contains an empirical application. Section \ref{sec:concl} concludes. Appendix \ref{app:proofs} contains proofs. Appendix \ref{app:shb} contains auxiliary theoretical results. Appendix \ref{app:emp} contains auxiliary empirical details.

\section{Setup}
\label{sec:notation}

 \subsection{Policy Learning Problem}
\label{sec:model}
Consider a population of individuals characterized by their potential outcomes in treated and untreated states, $Y(1), Y(0) \in \mathcal{Y} \subseteq \mathbf{R}$, and characteristics $X \in \mathcal{X} \subseteq \mathbf{R}^{d_X}$. A decision-maker (DM) aims to maximize the average welfare in the population by subjecting some individuals to treatment, depending on their observable characteristics $X$. That is, the DM chooses a treatment rule $G \in \mathcal{G} \subseteq 2^\mathcal{X}$ to maximize
\begin{equation} \label{eq:wfbstar}
W_G = \E[ Y(1) \bm{1}\{ X \in G \} + Y(0) \bm{1}\{ X \in G^c \}],
\end{equation}
where $G^c = \mathcal{X} \backslash G$ denotes the complement of $G$. The class of feasible treatment rules $\mathcal{G}$ may be \textit{ex ante} restricted for institutional reasons, such as transparency or non-discrimination in treatment, or practical reasons, such as computation and implementation.

We assume that the DM has access to experimental data that identifies $W_G$. The observable data vector $Z = (D, Y, X)$ contains the assigned treatment $D \in \{0, 1\}$, realized outcome $Y \in \mathcal{Y}$, and covariates $X \in \mathcal{X}$, so that  $Y = DY(1) + (1-D) Y(0)$ and $D \perp (Y(1), Y(0)) \,|\, X.$ The propensity score will be denoted by $\pi(x) = P(D = 1 \,|\,X = x)$. The conditional mean and variance functions of the potential outcomes are non-parametrically identified as $ m(d, x) = \E[Y(d) \,|\,X = x] = \E[Y \,|\, D = d, X = x]$ and $\sigma^2(d, x) = \Var(Y(d)\,|\,X = x) = \Var(Y \,|\,D = d, X = x)$, for  $d \in \{0, 1\}$, and the conditional average treatment effect (CATE) function as $\tau(x) = m(1, x) - m(0, x)$. As a result, the average welfare function is identified as $W_G = \E[m(0, X) + \bm{1}\{ X \in G \} \tau(X)]$ and can be non-parametrically estimated. To this end, the DM observes a random sample $(Z_i)_{i = 1}^N$ distributed i.i.d. $Z_i \sim P \in \mathbf{P}$, for a class of distributions $\mathbf{P}$  specified below.

The objects of interest throughout the paper are the maximum (or first-best, or optimal) welfare, denoted by
\begin{equation} \label{eq:max_welfare}
W_{G^*} = \max \limits_{G \subseteq \mathcal{G}} W_G,
\end{equation}
where $G^*$ denotes any policy attaining the maximum,\footnote{For simplicity, we assume that the maximum is well-defined.} and the corresponding welfare gain,
\begin{equation} \label{eq:welfaregain}
W_{G^*}^{\text{gain}} = W_{G^*} - W_{\varnothing},
\end{equation}
which is non-negative as long as the policy class $\mathcal{G}$ includes the status quo policy $\varnothing$ of not treating anyone.

\subsection{ Lower Confidence Bands} 
\label{sec:wgap}

In many settings, the DM would naturally be interested in lower confidence bands (LCBs) for the maximum welfare or the corresponding welfare gain. For example, consider a firm deciding whether to build a job-training center. Suppose the firm maximizes the net welfare subject to a ``safety'' constraint that the risk of false adoption (i.e., incurring negative welfare) must be below level $\alpha$, for some $\alpha \in (0,1)$. This leads to testing
\[
H_0:\; W_{G^*} \leq  0
\qquad\text{vs}\qquad
H_1:\;  W_{G^*} >  0.
\]
In such settings, LCBs are natural inputs to threshold decision rules \citep[see, e.g., Section 3.5 of][]{LehmannRomano}.

As another example, consider an online retailer deciding whether to offer a discount for a certain type of good to its customers. The retailer may first run a small-scale randomized experiment to explore whether any discount rule can lead to increase in profits. This corresponds to testing 
\[
H_0:\; W_{G^*}^{\text{gain}} = 0
\qquad\text{vs}\qquad
H_1:\; W_{G^*}^{\text{gain}} > 0,
\]
which is equivalent to comparing a \(100(1-\alpha)\%\) LCB for \(W_{G^*}^{\text{gain}}\) with zero.

The main input in the construction of LCBs is an estimator for the welfare function $W_G$. For each policy $G$, we can express $W_G = \E[\psi_G(Z)]$, where 
\begin{align}
\label{eq:dr}
\psi_G(Z) &= \bigg(m(1,X) + \dfrac{D}{\pi(X)} (Y - m(1,X)) \bigg)  \bm{1}\{ X \in G \} \\[2mm] 
&+ \bigg(m(0,X) +  \dfrac{1-D}{1-\pi(X)} (Y - m(0,X)) \bigg) \bm{1}\{ X \in G^c \} \nonumber
\end{align}
is the efficient, doubly robust, moment function \cite{Robins, Hahn98}. For suitable first-stage estimators $\widehat{m}(d, x)$ and $\widehat{\pi}(x)$, a regular semiparametrically efficient estimator $\widehat{W}_G$ can be constructed using cross-fitting, so that
\begin{equation} \label{eq:wg_asymp}
  \sqrt N(\widehat W_G - W_G) \;\Rightarrow^d \; \mathcal{N}(0, \sigma_G^2),
\end{equation}
where $\sigma_G^2 = \Var(\psi_G(Z))$. Given a significance level $\alpha \in (0,1)$, a $100(1-\alpha) \%$ LCB for $W_G$ can be formed as
 \begin{align}
\widehat{LCB}_{G} &= \widehat{W}_{G} - N^{-1/2} z_{1-\alpha} \widehat{\sigma}_{G},  \label{eq:lcbg} 
\end{align}
where  $z_{1-\alpha}$ is the $(1-\alpha)$ quantile of $\mathcal{N}(0,1)$ and $ \widehat{\sigma}_{G}$ is a consistent estimator of the asymptotic standard deviation $\sigma_G$. 

Since $W_{G} \leq W_{G^*}$, for any $G \in \mathcal{G}$, an LCB based on any suboptimal policy $G \in \mathcal{G}$ provides valid one-sided coverage for the optimal welfare, 
\begin{align}
\label{eq:lcbvalid}
P(\widehat{LCB}_G \le W_{G^*})
  \;&\ge\;P(\widehat{LCB}_G\le W_G)\to1-\alpha, \;\;\text{ as } N \rightarrow \infty.
\end{align}
As a result, $\widehat{LCB}_{G}$ can be meaningfully compared across distinct policies. As an ideal benchmark, we consider an LCB based on an infeasible efficient estimator of the welfare under a first-best policy, 
 \begin{align}
\widehat{LCB}_{G^{*}} &= \widehat{W}_{G^{*}} - N^{-1/2} z_{1-\alpha}\widehat{\sigma}_{G^{*}}. \label{eq:lcbg2} 
\end{align}
As discussed in the introduction, such LCB is a valid reference point even when the optimal policy is not unique. Since $\widehat{LCB}_{G^*}$ is based on an efficient estimator for $W_{G^*}$ and the standard deviation is rescaled by $N^{-1/2}$, one might expect that $\widehat{LCB}_{G^{*}}$ always be preferred to $\widehat{LCB}_{G}$ in large samples, for any suboptimal policy $G$. We show, however, that this is not the case. Given the direction of the intended comparison, considering an oracle LCB as a benchmark only strengthens our point. Of course, our recommended inference procedures in Section \ref{sec:moment} account for the first-best policy being unknown.

\subsection{Asymptotic Criterion for LCB ranking }
\label{sec:lcbgap}

To compare the candidate LCBs, we consider the \textit{LCB gap}, defined as
\begin{align}
\label{eq:deltag}
\Delta_G &= N^{-1/2} z_{1-\alpha} (\sigma_{G^{*}}-\sigma_G)  -(W_{G^{*}} - W_G).
\end{align}
A positive sign of $\Delta_G$ indicates that the policy $G$ is nearly optimal yet the corresponding welfare is substantially more precisely estimated. Consequently, the corresponding ${LCB}_{G}$ may be preferred to ${LCB}_{G^{*}}$ in large samples. 

The motivation for studying LCB gap comes from a local asymptotic approximation along smooth parametric sub-models, standard in the semi-parametric efficiency theory. To elaborate, let $\mathbf{P}$ denote the class of all admissible distributions of the data. Consider a distribution $P_0 \in \mathbf{P}$ such that $W_{G^*(P_0)} = W_{G}$ for $G \ne G^*(P_0)$. Let $T(P_0)$ denote the tangent space at $P_0$,\footnote{See \citet{Hahn98} for the derivation of $T(P)$ in the present setting.} and $P_{N, h} = P_{1/\sqrt{N}, h}$, for $h \in T(P_0)$, be a sequence of distributions following a smooth parametric submodel $\{t \mapsto P_{t, h}\} \subseteq \mathbf{P}$. Denote
\[
\begin{array}{cl}
\mu(h) &= \sqrt{N}(W_{G^*(P_{N, h})} - W_{G});\\[3mm]
s(h) &= \sigma_{G^*(P_{N, h})} - \sigma_{G},
\end{array}
\] 
where the dependence of $\mu(h)$ and $s(h)$ on $N$ is suppressed for notational convenience, and note that
\[
\Delta_{G}(P_{N, h}) = N^{-1/2}(z_{1-\alpha} s(h) - \mu(h)).
\]
The assumed regularity of $\widehat{W}_{G}$, consistency of $\widehat{\sigma}_{G}$, and contiguity of $P_{N, h}$ with respect to $P_0$ imply that, under $P_{N, h}$,
\[
\sqrt{N}(\widehat{LCB}_{G} - \widehat{LCB}_{G^*(P_{N, h})}) \;\;\; \Rightarrow_d \;\;\; \mathcal{N}(z_{1-\alpha} s(h) - \mu(h), \sigma_{\Delta}^2(P_0)), 
\]
 for some $\sigma_{\Delta}^2(P_0) \geqslant 0$. That is, the distribution of $\sqrt{N}(\widehat{LCB}_{G} - \widehat{LCB}_{G^*(P_{N, h})})$ under any sequence of ``perturbations'' $P_{N, h}$ of $P_0$, is determined by $z_{1-\alpha}s(h) - \mu(h) = \sqrt{N} \Delta_{G}(P_{N, h})$. Moreover, under further regularity conditions, 
\[
\E[\widehat{LCB}_{G} - \widehat{LCB}_{G^*(P_{N, h})}] = \Delta_{G}(P_{N, h}) + o(1),
\]
so the LCB gap can be interpreted as a large-sample analog to the difference of expected LCBs.\footnote{An ideal way to rank LCBs is in terms of the first-order dominance; See \citet{Lehmann1959}. Unfortunately, since distinct policies typically result in LCBs with distinct large-sample variances, this criterion does not apply in a Gaussian limit. A natural alternative is to compare LCBs in terms of their expected values, as suggested, e.g.,  in \citet{LeonHarter}. While the exact expectations may not exist without further restrictions or be distorted by the biases in first-stage estimators, their large sample analogs remain tractable.} For these reasons, we consider the LCB gap in the formal results below.

\section{First-Order Dominance}
\label{sec:example}

In this section, we give an example of a model in which the welfare-precision trade-off is of the first order, and discuss the implications of this phenomenon. We focus on welfare throughout, but similar considerations apply to welfare gain. See Remark \ref{rm:welfaregain} for the details.
 
\subsection{The Data Generating Process}
\label{sec:dgp}

First, we specify a suitable DGP for $(Y(1), Y(0), D, X)$. It suffices to specify the marginal distribution of $X$, the propensity score, and the conditional distributions of $Y(1) \mid X$ and $Y(0) \mid X$. Let $X$ be a binary covariate distributed as
\begin{equation*}
    P (X=1) =p; \quad P (X=0) =1-p, \quad \text{ for some } p \in (1/4, 3/4).
\end{equation*}
Denote the propensity score by
$$
P (D=1 \mid X=1) = \pi(1); \quad P (D=1 \mid X=0) = \pi(0), \quad \text{ for some } \pi(1), \pi(0) \in (1/4, 3/4).
$$
Let $F(\mu, \sigma^2)$ be any distribution with mean $\mu$ and variance $\sigma^2$. Suppose the potential outcomes are distributed as 
\begin{align}
\label{eq:line2}
Y(1)\mid X=1 &\sim F\bigl(\tfrac12-\epsilon,\,1\bigr); 
\qquad
Y(1)\mid X=0 \sim F\bigl(\tfrac12,\,1\bigr); 
\\ \label{eq:dgp}
Y(0) \mid X=1 &\sim F\bigl(\tfrac12,\,10 \bigr); 
 \qquad \quad \,
Y(0) \mid X=0 \sim F\bigl(\tfrac12-\epsilon,\,10 \bigr),
\end{align}
where $\epsilon \in (0, 1/2)$ is a vanishing sequence to be specified. Since we focus on the average welfare, the joint distribution of $(Y(1), Y(0)) \,|\, X$ is immaterial, so we leave it unspecified\footnote{ With variance parameters $\sigma^2(1,1) = 1/4, \quad \sigma^2 (1,0)=1, \quad \sigma^2(0,1) = 200, \quad \sigma^2(0,0) = 1$, the statement holds for all sample sizes exceeding  $1745$. For the variances in the main text, the minimal cutoff sample size $N$ is approximately $6 000$.}

Simple algebra shows that the CATE function takes the form
$$
\tau(1) = -\epsilon < 0; \;\;\;\;\;\; \tau(0) =\epsilon>0,
$$
the unique first-best policy is 
\begin{align}
\label{eq:fbb}
G^{*} = \{0\},
\end{align}
and the corresponding welfare is
\begin{align}
\label{eq:fbb2}
W_{G^{*}} = 1/2 \cdot p + 1/2 \cdot (1-p)=1/2.
\end{align}
In addition, consider  the ``treat everyone'' policy, $G=\mathcal{X}$ whose welfare is
\begin{align}
\label{eq:fbb3}
W_{\mathcal{X}} = (1/2 - \epsilon) \cdot p + 1/2 \cdot (1-p) = 1/2 - \epsilon p.
\end{align}
Note that the welfare gap between the two policies scales linearly with $\epsilon$ 
 \begin{align}
\label{eq:puzzlebias}
0 \leq W_{G^{*}}-W_{\mathcal{X}} \leq \epsilon p.
\end{align}
while the standard deviation gap does not depend on $\epsilon$,
\begin{align}
\label{eq:sigmalb2}
\sigma_{G^{*}}-\sigma_{\mathcal{X}} &> p. 
\end{align}

\subsection{Estimators and Lower Confidence Bands}
\label{sec:est_lcb}

Since $X$ is binary, the average welfare under any fixed policy $G$ can be efficiently estimated using the regression-adjusted estimator. For each \((d,x) \in \{0, 1\}^2\), denote
\begin{align}
\label{eq:ndx}
N_{dx} = \sum_{i=1}^N \bm{1}\{ D_i = d \} \bm{1}\{ X_i =x \},
\end{align}
and define the esitmators
\begin{equation} \label{eq:mdx}
\begin{array}{cl} 
\widehat \pi(x) &= \dfrac{N_{1x}}{N_{1x} + N_{0x}};\\[5mm] 
\widehat m(d,x)  &= \dfrac{\sum_{i=1}^N Y_i \cdot \bm{1}\{ D_i = d \} \bm{1}\{ X_i =x \}}{N_{dx}+1}, 
\end{array}
\end{equation}
where one is added to the denominator throughout to prevent division by zero.\footnote{This step introduces bias of order $O(N^{-1})$ which is negligible for a sufficiently large sample. An alternative is to work with unadjusted denominators on the event where both of them are strictly positive.} Recalling from \eqref{eq:fbb} that $G^* = \{0\}$, the first-best welfare is estimated as
\begin{align}
\widehat {W}_{G^{*}} =  \widehat m(0, 1) \cdot  \widehat p +  \widehat m(1, 0) \cdot (1- \widehat p), \label{eq:oracle2}  
\end{align}
where $\widehat p  = \sum_{i=1}^N X_i/N$. Similarly,  
\begin{align}
\widehat {W}_{\mathcal{X}} &= \widehat m(1, 1) \cdot  \widehat{p} + \widehat m(1, 0) \cdot (1-\widehat{p}). \label{eq:m0} 
\end{align}
The mean squared errors of the two estimators, with respect to $W_G^*$, are given by
\begin{equation} \label{eq:mse1} \arraycolsep=1.pt
    \begin{array}{cl}
       MSE (\widehat{W}_{\mathcal{X}})  &= \E [(\widehat{W}_{\mathcal{X}} - W_{G^{*}})^2];  \\[2mm]
       MSE(\widehat W_{G^{*}})  &= \E [(\widehat W_{G^{*}} - W_{G^{*}})^2].
    \end{array}
\end{equation} 
The asymptotic variances of $\widehat{W}_G$, for $G \in \{G^*, \mathcal{X}\}$, can be estimated as
$$
\widehat \sigma^2_G = \dfrac{1}{N} \sum_{i=1}^N (\widehat \psi_G(Z_i) - \widehat {W}_G)^2, 
$$
where $\widehat{\psi}_G(Z_i)$ is obtained by plugging the estimated propensity score  and regression functions from \eqref{eq:mdx} in \eqref{eq:dr}. The corresponding LCBs are obtained as
\begin{align}
\widehat{LCB}_{\mathcal{X}} &= \widehat{W}_{\mathcal{X}} -  N^{-1/2} z_{1-\alpha} \widehat{\sigma}_{\mathcal{X}}, \label{eq:lcb0} \\[2mm]
\widehat{LCB}_{G^{*}} &= \widehat{W}_{G^{*}} - N^{-1/2} z_{1-\alpha} \widehat{\sigma}_{G^{*}}.  \label{eq:lcbfb} 
\end{align}
Following the discussion of Section \ref{sec:notation}, we compare ${LCB}_{\mathcal{X}}$ and ${LCB}_{G^{*}}$ in terms of LCB gap 
\begin{align}
\label{eq:LCBgap}
\Delta_{\mathcal X}
&= \frac{z_{1-\alpha}}{\sqrt{N}}\,\bigl(\sigma_{G^*} - \sigma_{\mathcal X}\bigr)
\;-\; \bigl(W_{G^*} - W_{\mathcal X}\bigr).
\end{align}

\subsection{First-Order Dominance}
\label{sec:mainr}

Our first main result shows that $\widehat{W}_{\mathcal{X}}$ dominates $\widehat{W}_{G^{*}}$ in terms of MSE, and the respective LCB gap is positive.

\begin{proposition}[First-Order Dominance]
\label{thm:main1}
For all $N$ large enough, for the DGP \eqref{eq:dgp} and estimators \eqref{eq:oracle2} and \eqref{eq:m0}, the following statements hold:
 \begin{enumerate}
     \item Both MSEs in \eqref{eq:mse1} are finite and
\begin{align}
\text{MSE} (\widehat{W}_{\mathcal{X}})  < \text{MSE} (\widehat{W}_{G^{*}});\label{eq:mse}
\end{align}
\item For any significance level $\alpha \in (0,1)$, there is a constant $C_{\alpha}>0$ such that
\begin{align}
\Delta_{\mathcal{X}}  >  C_{\alpha} N^{-1/2}. \label{eq:mainstatement}
\end{align} 
 \end{enumerate}
 
\end{proposition}

Proposition \ref{thm:main1} makes three points. First, the trade-off between welfare and precision may be first-order. As a result, suboptimal policies may yield better point estimates and tighter, on average, lower confidence bands for the optimal welfare. That is, the first-best policy --- the policy that is best to implement --- may differ from the policy whose estimated welfare is best to report.\footnote{DGPs with treatment effects vanishing at the $N^{-1/2}$ rate have been employed to obtain a meaningful limiting experiment \citep{HiranoPorter2009} or establish minimax rates for expected regret \citep{KitagawaTetenov, AtheyWager2}. In this paper, we use DGPs with similar conditional means and carefully chosen variances to establish a lower bound on the LCB gap.} Similar observations apply to inference on partially-identified parameters, as we further discuss in Remark \ref{rm:vanish2}.

Second, there is a distinction between the two-sided and one-sided inferential objectives.  In the two-sided case, the bias typically must vanish faster than the standard deviation to ensure valid coverage of the confidence intervals. In the one-sided case, coverage remains valid as long as the direction of the bias matches the direction of the confidence band, which allows bias and variance to be potentially of the same order.  Proposition \ref{thm:main1} gives a concrete, empirically relevant example of this distinction\footnote{The one-sided dominance result echoes findings in one-sided nonparametric inference: in adaptive tests and multiscale procedures, directed smoothing bias can be exploited to lower variance while preserving size \citep{DumbgenSpokoiny2001, Armstrong2015}. Our setting differs in the target parameter (optimal welfare rather than a function at a point) and mechanism (policy-induced bias $W_{G^{*}} - W_G$ rather than smoothing bias).   }.

Third, efficiency arguments in near non-regular settings may be problematic. For each $\epsilon>0$, the oracle efficient estimator $\widehat{W}_{G^*}$ attains the semiparametric efficiency bound \citep{LuedtkeLaan}, but in the limit, $\epsilon=0$, the optimal welfare is a non-regular parameter, and semiparametric efficiency bounds do not apply \citep{HiranoPorter2012}. Proposition \ref{thm:main1} demonstrates that, for distributions within a $N^{-1/2}$-neighborhood of $\epsilon=0$ (excluding zero), $\widehat{W}_{\mathcal{X}}$ dominates $\widehat{W}_{G^{*}}$ in terms of MSE, for all $N$ large enough. Thus, the familiar notion of efficiency fails not only at the point of non-regularity, but already in a $N^{-1/2}$-neighborhood around it.

\begin{remark}[Implications for welfare gain]
\label{rm:welfaregain}
The above example could be modified to obtain a first-order dominance statement for the welfare gain in \eqref{eq:welfaregain}. Consider the DGPs 
\begin{align}
Y(1)\mid X=1 &\sim F\bigl(\tfrac12 - \epsilon,\;1\bigr), 
& 
Y(1)\mid X=0 &\sim F\bigl(\tfrac12,\;10\bigr), 
\label{eq:dgp1}\\
Y(0)\mid X=1 &\sim F\bigl(\tfrac12,\;1\bigr), 
& 
Y(0)\mid X=0 &\sim F\bigl(\tfrac12 - \epsilon,\;10\bigr).
\label{eq:dgp3}
\end{align}
where asymptotic variance is small for $X=1$ and large for $X=0$. Let $G^{*} = \{0\}$ be the optimal policy and $G = \{1\}$ be the suboptimal policy. Simple algebra shows that the welfare gap and variance gap satisfy
 \begin{align}
\label{eq:puzzlebias2}
W^{\text{gain}}_{G^{*}}-W^{\text{gain}}_{G}  &\leq \epsilon,  \\
\label{eq:puzzlebias3}
 (\sigma^{\text{gain}}_{G^{*}})^2-(\sigma^{\text{gain}}_{G})^2_{}  &> 7.
\end{align}
As a result, an analog of \eqref{eq:mainstatement} holds for the LCB gap for welfare gain.  \qed 
 \end{remark}

\begin{remark}[Redundant moment inequalities]
\label{rm:vanish2}
The above discussion applies to inference for partially identified parameters. For example, consider the setting of Section \ref{sec:notation} with binary potential outcomes and unconditional treatment exogeneity, i.e. $(Y(1), Y(0), X) \perp D$. The share of ``always-takers'', $\theta=P(Y(1) = Y(0) = 1),$
can be bounded from above by either $\delta_1 = P(Y=1 \mid D=0)$ or $\delta_2 = \E [\min (P(Y=1 \mid D=1,X), P(Y=1 \mid D=0,X))]$. By Jensen's inequality, $\delta_2$ gives a tighter bound than $\delta_1$. A $100(1-\alpha) \%$ Upper Confidence Band (UCB) for $\theta$ can be formed using either of the two bounds
$$
\widehat{UCB}_j = \widehat{\delta}_j + N^{-1/2} z_{1-\alpha} \widehat{\sigma}_j,
$$
where $\widehat \sigma_j$ are consistent estimators of the asymptotic standard deviations $\sigma_j$ of $\hat{\delta}_j$, for $j = 1, 2$. Similar to Proposition \ref{thm:main1}, there exist DGPs such that
\begin{align}
\label{eq:mainstatement3}
\delta_2+N^{-1/2} z_{1-\alpha} \sigma_2>\delta_1+N^{-1/2} z_{1-\alpha} \sigma_1,
\end{align}
for all $N$ large enough. As a result, a UCB based on a non-sharp bound first-order dominates its sharp counterpart in terms of the average length. In other words, inference based on a sharp bound may be less informative in finite samples. \qed 
\end{remark}

\section{Margin Condition and Higher-Order Dominance}
\label{sec:mainr2}

Next, we investigate whether the conclusions of Proposition \ref{thm:main1} carry over when the model is restricted by the following additional assumptions.  
  
\begin{assumption}[Regularity]
\label{ass:overlap}
(i) The propensity score $\pi(x)$ satisfies $\kappa < \pi(x) < 1-\kappa$, for almost all $x \in \mathcal{X}$, for some $\kappa \in (0, 1/2)$; (ii) The outcome is bounded so that $P (|Y| \leq M/2) = 1$, for some $M < \infty$.  
\end{assumption}

\begin{assumption}[Margin Condition]
\label{ass:margin}
For some  $\eta \in (0,M)$  and $\delta \in (0,\infty)$, 
\begin{align}
\label{eq:margin}
P(|\tau(X)| < t) \leq  (t/\eta)^{\delta}, \quad \forall t \in [0, \eta).
\end{align}
\end{assumption}

 Assumption \ref{ass:overlap} imposes regularity conditions common in the policy learning literature \citep[see, e.g., ][]{KitagawaTetenov,MbakopTabord}.
Assumption \ref{ass:margin} is the margin condition of \citet{Tsybakov}. In addition to requiring uniqueness of the first-best policy, it controls the intensity with which $\tau(X)$ concentrates in a neighborhood of zero. When the optimal policy is unique, the existence of suitable values of $\delta$ and $\eta$ is a matter of mild regularity conditions. For example, if $|\tau(X)|$ is continuous and has a density bounded at zero, then \eqref{eq:margin} holds for any $\delta<1$ with $\eta$ small enough. If $\tau(X)$ has finite support and $P(\tau(X) = 0) = 0$, then \eqref{eq:margin} holds for any $\delta > 0$ and a sufficiently small $\eta$.

The sequence of DGPs in Proposition \ref{thm:main1} can be chosen to satisfy Assumption \ref{ass:overlap}, but it fails to satisfy Assumption \ref{ass:margin} with uniform lower bounds on $\eta$ and $\delta$. As we show below, this is precisely what drives the first-order dominance phenomenon. To state the formal result, we assume that any $G \subseteq \mathcal{X}$ is feasible.\footnote{The upper bound in Proposition 2 holds for all $G \subseteq \mathcal{X}$, so it applies to any restricted class $\mathcal{G}$ as well. The lower bound holds within restricted classes $\mathcal{G}$ as long as they include threshold policies based on each covariate.} Proposition \ref{thm:equivalence} below characterizes the order of magnitude of the worst-case LCB gap $\Delta_G$ over all policies $G \subseteq \mathcal{X}$.

\begin{proposition}[Higher-Order Dominance]
\label{thm:equivalence}
Let $\mathbf{P}$ denote the class of DGPs obeying Assumptions \ref{ass:overlap}--\ref{ass:margin} for some $ 0 < \underline{\delta} \leq \delta \leq \overline{\delta} < \infty$, $\eta = \eta(\delta) > 0$, and $\inf_{x \in \mathcal{X}, d \in \{1,0\}} \sigma^2(d, x) \geq \underline{\sigma}^2>0$. There exist constants $0<\underline C<\overline C<\infty$, depending on $(M,\kappa, \underline{\delta}, \overline{\delta}, \underline{\sigma})$, such that
\begin{equation}
	\label{eq:mainresultupper}
\underline{C} N^{-(1 + \underline{\delta})/2} \leq \sup_{P \in \mathbf{P}} \sup_{G \subseteq \mathcal{X} } \Delta_G \leq \overline{C} N^{-(1 + \underline{\delta})/2}.
\end{equation}

\end{proposition}

Proposition \ref{thm:equivalence} demonstrates that once uniform lower bounds on $\delta$ and $\eta$ are imposed, no suboptimal policy $G$ can lead to first-order dominance in the sense of Proposition \ref{thm:main1}. The smaller the value of $\delta$, the more  $\tau(X)$ concentrates near zero, the looser the upper bound in \eqref{eq:mainresultupper}.  In the limit, $\delta=0$, which corresponds to failure of the margin condition, the lower bound in \eqref{eq:mainresultupper} recovers the first-order dominance result \eqref{eq:mainstatement}. In the absence of uniform bounds on the margin parameters, Propositions \ref{thm:main1} and \ref{thm:equivalence} imply that the first-best welfare may not be the optimal, or relevant, inferential target. The following remarks discuss testable implications of the margin condition and possible testing procedures, as well as further connections with the literature.

\begin{remark}[Testing uniqueness of the optimal policy]
\label{rm:ftest} 
Let $X$ be a discrete covariate taking $J$ distinct values with positive probabilities. Then, the conditional average treatment effect reduces to a vector $(\tau(j))_{j=1}^J$. The first-best policy is non-unique if (and only if) $\tau(j) = 0$ for some $j \in \{1, 2, \dots, J\}$. The null hypothesis 
\begin{align}
\label{eq:h0null}
H_0: \exists j:  \tau(j) = 0
\end{align}
is a union of $J$ simple hypotheses $H_{0j}: \tau(j)=0$. Then, letting $R_j$ denote the rejection region for testing $H_{0j}$, the test with a rejection region
$$
R=\cap_{j=1}^J R_j,
$$
is valid for $H_0$, although typically conservative \citep[see, e.g.,][]{Berger1997}. \qed 
\end{remark}

\begin{remark}[Testing the margin assumption]
\label{rm:vanish}
In the general case where both discrete and continuous covariates are present,   Assumption \ref{ass:margin} is no longer equivalent to uniqueness of the optimal policy. 
We describe a testable implication that we find empirically relevant in Section \ref{sec:jc}.  Let $P(G^* \triangle G)$ denote the share of people treated differently under the optimal policy $G^*$ and an alternative $G$. This share links welfare and standard deviation gaps. Specifically, the welfare gap is lower bounded as 
\begin{align}
W_{G^*} - W_{G} \geq C_1P (G^{*} \triangle G)^{1+\frac{1}{\delta}}  \label{eq:marginmain2} 
\end{align}
for $C_1 = C_1(\delta)=\eta\delta(\frac{1}{1+\delta})^{1 + \frac{1}{\delta}} > 0$ \citep{Tsybakov}. Given a lower bound $\underline{\delta}>0$ and fixing $\eta>0$, consider  a null hypothesis  $H_0: \delta \geqslant \underline{\delta}$. Since both functions $\delta \rightarrow C_1(\delta)$ and $\delta \rightarrow  c^{1+\frac{1}{\delta}}$  are increasing in $\delta$, the lower bound \eqref{eq:marginmain2} on welfare gap implies that, for any policy $G$, 
\begin{align}
 C_1(\underline{\delta}) P(G^* \triangle G)^{1 + \frac{1}{\underline{\delta}}} - (W_{G^*} - W_G) \leq 0. \label{eq:marginmain3}
\end{align}
In particular, if the welfare gap $W_{G^*} - W_G$ of some policy $G$ vanishes with sample size, the share of people treated differently under $G$ and $G^*$, must vanish, too. Existing methods from the moment inequality literature, such as \citet{CLR} and \citet{CherNeweySantos}, can then be applied to construct a test. Pursuing this formally is left for future work.\footnote{The lower bound in~\eqref{eq:marginmain2} plays a role analogous in spirit to the polynomial minorant condition used in partial identification literature, e.g., Condition~C.2 in \citet{CHT}, Condition~V in \citet{CLR}, and Assumption~4.2 in \citet{Armstrong2014}. In its general form, this condition relates the difference in the criterion function to the distance metric on the parameter of interest.  In policy learning settings, the decision set $\mathcal{G}$ is a collection of partitions of covariate space.  In both \citet{CLR} and \citet{KitagawaTetenov}, this condition is imposed to tighten convergence guarantees for the proposed estimators.  In contrast to prior work, this paper uses the (failure of) margin assumption to motivate the use of suboptimal policies for constructing lower confidence bands for welfare. } \qed 
\end{remark}

\begin{remark}[Implications for debiased inference]
\label{rm:margin}
 Propositions \ref{thm:main1} and \ref{thm:equivalence} imply that sharp bounds may not be optimal, or relevant, inferential targets in the absence of uniform margins, highlighting the tightness of this condition in the context of covariate-assisted bounds; see \citet{kallus2020assessing, kallus2022whats,kallus2022treatment, levis2023covariateassisted, SemSupp2, Semenova2024}. We expect this insight to imply the tightness of the margin condition in other settings, such as support function analysis \cite{CCMS} and algorithmic fairness \cite{liu2025}, and other policy-relevant metrics. \qed
\end{remark}

\section{Robust Testing Procedures}
\label{sec:moment}

In this section we discuss testing procedures that address the welfare-precision trade-off and remain valid regardless of the margin assumption. Let $\mathcal{G}_{\text{test}} \subseteq \mathcal{G}$ be a class of policies, which, based on economic intuition, may contain a good lower bound for the optimal welfare. We look for a LCB of the form
	\begin{equation} \label{eq:lcb_generic}
	\widehat{LCB} = \max \limits_{G \in \mathcal{G}_{\text{test}}} \left\{ \widehat{W}_G - \hat{c}_{\alpha}  \frac{\hat{\sigma}_G}{\sqrt{N}} \right\},
	\end{equation}
	where $\hat{c}_\alpha$ is as small as possible to guarantee the desired coverage. We show that such LCB naturally arise from testing moment inequalities, which allows to use a host of existing testing procedures. Our results take the form of finite-sample algebraic identities, so the coverage properties of the resulting LCBs are inherited from validity of the underlying tests. The latter relies only on the uniform CLT-type assumptions and holds regardless of the margin condition.  We refer the reader to \citet{CLR} and \citet{canay2017practical} for the details. 

\subsection{Lower Confidence Bands via Testing Moment Inequalities}
\label{sec:lcb_testing_inequalities}

Suppose $\mathcal{G}_{\text{test}}$ is finite (potentially growing with sample size). Let $\theta = W_{G^*}$ denote the parameter of interest, and consider testing
\begin{equation} \label{eq:h0}
	H_{\theta}:  W_{G} - \theta \leq 0, \text{ for all } G \in \mathcal{G}_{\text{test}}.
\end{equation}
Suppose the estimator \((\widehat W_G)_{G\in\mathcal{G}_{\text{test}}} \) for \( (W_G)_{G\in\mathcal{G}_{\text{test}}}\) satisfies
\begin{equation} \label{eq:clt}
\bigl(\sqrt{N}(\widehat W_G - W_G)\bigr)_{G\in\mathcal{G}}
\;\Rightarrow^d\;
\mathcal{N}\bigl(0,\Sigma\bigr),
\end{equation}
for a positive definite covariance matrix $\Sigma= (\Sigma_{G_1G_2})_{G_1, G_2 \in \mathcal{G}}$, and a consistent estimator $\widehat{\Sigma}$ is available. A test for \eqref{eq:h0} can then be constructed as
\begin{equation} \label{eq:test}
\hat{\phi}_N(\theta) = \bm{1} \left( \hat{T}_N(\theta) > \hat{c}_{\alpha}(\theta)  \right),
\end{equation}
with, e.g., the maximum test statistic
\begin{equation}
\widehat{T}_N(\theta) = \max \limits_{G \in \mathcal{G}_{\text{test}}} \frac{\sqrt{N}(\widehat{W}_G - \theta)}{\hat{\sigma}_{G}},
\end{equation}
where $\hat{\sigma}_{G} = (\widehat{\Sigma}_{GG})^{1/2}$ and $\hat{c}_{\alpha}(\theta)$ is suitable a critical value. A common computationally simple choice is the least-favorable critical value, corresponding to 
\begin{equation} \label{eq:lf-max}
\hat{c}_{\alpha, \max}^{\text{LF}} = \hat{Q}_{1-\alpha}\left( \max \limits_{G \in \mathcal{G}_{\text{test}}} \frac{\sqrt{N}(\widehat{W}_G - W_G)}{\hat{\sigma}_G} \right),
\end{equation}
where the quantile can be estimated using bootstrap or Gaussian approximation.

 Given the direction of the inequalities in \eqref{eq:h0}, the set of all values of $\theta$ for which the test in \eqref{eq:test} does not reject, $\{\theta \in \mathbf{R}: \hat{\phi}_N(\theta) = 0\}$, provides a LCB for $W_{G^*}$. For the least-favorable critical value, the test compares the value of a partially linear decreasing function of $\theta$ with a constant, which allows to obtain a simple closed form for the LCB.

\begin{proposition}[LCB by test inversion]
\label{prop:max}
The LCB obtained by inverting a test in \eqref{eq:test} with the least-favorable critical value in \eqref{eq:lf-max} is given by
\begin{equation} \label{eq:lcb_max_lf}
\widehat{LCB}_{\max}^{LF} =
\max_{G\in\mathcal G_{\text{\normalfont test}}}
\Bigl\{
\widehat W_G - \hat{c}_{\alpha, \max}^{LF}\frac{\widehat \sigma_G}{\sqrt{N}}\,
\Bigr\}.
\end{equation}
\end{proposition}

Intuitively, the above procedure corresponds to constructing a candidate LCB for $W_{G^*}$ using each suboptimal policy $G \in \mathcal{G}_{\text{test}}$ separately and taking the shortest one, thus explicitly resolving the welfare-precision trade-off. The least-favorable critical value $\hat{c}_{\alpha, \max}^{LF}$ ensures that the resulting LCB has the desired coverage, but it essentially assumes that all of the moment inequalities in \eqref{eq:h0} are binding, which may be too conservative. The critival value can be reduced using moment selection procedures, such as the Generalized Moment Selection (GMS) of \citet{andrews2010inference}, or pre-testing, as in \citet{romano2014practical}. Although both procedures perform well in practice, we focus on GMS because it allows for closed-form test inversion.

The critical value for the GMS procedure is computed as follows. Define the set of inequalities that are ``close to binding,''
\[
I_{N}(\theta) = \left\{G \in \mathcal{G}_{\text{test}}: \frac{\sqrt{N}(\hat{W}_G - \theta)}{\hat{\sigma}_G}  > -\kappa_{N} \right\},
\]
where $\kappa_N > 0$ is a sequence of tuning parameters such that $\kappa_N \to \infty$ and $\kappa_N / \sqrt{N} \to 0$, for example, $\kappa_N = \sqrt{\log N}$. Then, the GMS critical value is
\begin{equation} \label{eq:crit_val_gms}
	\hat{c}_{\alpha, \max}^{GMS}(\theta) = \hat{Q}_{1-\alpha}\left( \max \limits_{G \in I_N(\theta)}  \frac{\sqrt{N}(\widehat{W}_G - W_G)}{\hat{\sigma}_G} \right),
\end{equation}
where the quantile can be estimated using bootstrap or Gaussian approximation. As $\theta$ increases, the set $I_N(\theta)$ shrinks, so $\hat{c}_{\alpha, \max}^{GMS}(\theta)$ is a decreasing step-function of $\theta$. Thus, the test in \eqref{eq:test} with the GMS critical value compares a partially linear decreasing function of $\theta$ with a step-function. Since there may be multiple intersections, the confidence region obtained by test inversion may not be convex, although it can be shown that the probability of such an event approaches zero as $N$ increases. In the statement below, we conservatively define the LCB starting from the lowest intersection point.

\begin{proposition}[LCB by test inversion with GMS] \label{prop:max_gms} The LCB obtained by inverting the test in \eqref{eq:test} with the critical value \eqref{eq:crit_val_gms} can be computed as follows. For $j \in \{1, \dots, |\mathcal{G}_{\text{test}}|\}$, let $t^{(j)}$ denote the $j$-th largest value among $\widehat{W}_G + \kappa_N \widehat{\sigma}_G / \sqrt{N}$, and set $t^{(|\mathcal{G}_{\text{test}}| + 1)} = -\infty$. Let $I^{(j)} \subseteq \mathcal{G}_{\text{test}}$ collect the policies $G$ corresponding to $t^{(1)}, \dots, t^{(j)}$, and $\widehat{c}_{\alpha}^{(j)}$ be computed as in \eqref{eq:crit_val_gms} with $I^{(j)}$ instead of $I_N(\theta)$. Denote $ \hat{\theta}^{(j)} = \max_{G \in \mathcal{G}}(\widehat{W}_G - \widehat{c}^{(j)}_{\alpha} \widehat{\sigma}_G /\sqrt{N}).$ Then,
\begin{align}
\label{eq:lcb_max_gms}
\widehat{LCB}_{\max}^{GMS} = \min\{\theta^{(j)}: t^{(j)} \geqslant \hat{\theta}^{(j)} > t^{(j+1)}\}.
\end{align}
\end{proposition} 

The LCB in \eqref{eq:lcb_max_gms} uses a weakly smaller critical value than the LCB in \eqref{eq:lcb_max_lf}, so it is always shorter. Yet, the two LCBs are uniformly valid over the same set of distributions. Taken together, self-normalization of the test statistic and a moment selection procedure allow to resolve the welfare-precision trade-off while ensuring that the resulting LCB is robust to violations of the margin condition.

\subsection{Lower Confidence Bands via Intersection Bounds}

Inference methods for intersection-bounds-type parameters, such as $\max_{G \in \mathcal{G}}W_G$, have been introduced by \citet{CLR} (CLR for short). The authors pointed out that inference based on the plug-in estimator $\max_{G \in \mathcal{G}}\widehat{W}_G$ may be distorted for two reasons: upward bias and large differences in precision of estimates $\widehat{W}_G$ across $G \in \mathcal{G}$. To address these issues, they introduced a ``precision corrected'' LCB of the form \eqref{eq:lcb_generic} and proposed a different moment selection device, tailoring the analysis to an infinite number of intersection parameters (i.e., infinite $\mathcal{G}$). In what follows, we derive a new duality result between the procedure of CLR and test inversion in the spirit of Section \ref{sec:lcb_testing_inequalities} and use it to obtain a computationally simpler LCB.

In the preceding section, to find a good lower bound on $W_{G^*}$, we restricted attention to policies in the test class $\mathcal{G}_{\text{test}} \subseteq \mathcal{G}$. A better lower bound may potentially be obtained by taking convex combinations of $(W_G)_{G \in \mathcal{G}_{\text{test}}}$, which is equivalent to randomizing over $G \in \mathcal{G}_{\text{test}}$. Specifically, let $\Lambda = \bigl\{\lambda\in \mathbf{R}^{|\mathcal{G}_{\text{test}}|}_+:\,\mathbf{1}'\lambda=1\bigr\}$, where $\mathbf{1}=(1,\ldots,1)'\in \mathbf{R}^{|\mathcal{G}_{\text{test}}|}$, denote the probability simplex, $W_{\text{test}} = (W_G)_{G \in \mathcal{G}_{\text{test}}}$ collect the test policies into a finite vector, and $\widehat{W}_{\text{test}} = (\widehat{W}_G)_{G \in \mathcal{G}_{\text{test}}}$ denote the corresponding estimator vector. Each $\lambda \in \Lambda$ yields a lower bound  $\lambda'W_{\text{test}}\leq W_{G^*}$ for the optimal welfare. Therefore, following CLR, we look for a LCB of the form 
\begin{equation} \label{eq:lcbqlrmain}
\widehat{LCB}_{\text{mix}} = \max_{\lambda\in\Lambda} \left\{ \lambda^{\prime} \widehat{W}_{\text{test}} - \hat{c}_{\alpha} \frac{\sqrt{\lambda'\,\widehat\Sigma\,\lambda}}{\sqrt{N}} \right\},
\end{equation}
where the critical value $\hat{c}_{\alpha}$ is chosen to ensure correct coverage. 

A version of CLR's procedure calibrates $\hat{c}_{\alpha}$ by approximating the supremum of the self-normalized Gaussian process $(\sqrt{N}(\lambda'\widehat{W}_{\text{test}} - \lambda'W_{\text{test}}) / (\lambda'\widehat{\Sigma}\lambda)^{1/2} )_{\lambda \in \Lambda}$ in simulations, which can be computationally heavy. We replace that step with a finite-dimensional convex program using convex duality. We show that for any vector $T$ and positive definite matrix $\Sigma$,
\begin{equation}
\label{eq:gp-to-convex}
\max_{\lambda\in\Lambda}\frac{\lambda'T}{\sqrt{\lambda'\Sigma\,\lambda}} =
\Big(\,\min_{t \in \mathbf{R}^{|\mathcal{G}_{\text{test}}|}_{-} }\ (T-t)'\,\Sigma^{-1}\,(T-t)\,\Big)^{1/2},
\end{equation}
Consequently, an LCB of the form \eqref{eq:lcbqlrmain} actually arises from inverting a test for \eqref{eq:h0} using the so-called Quasi-Likelihood-Ratio (QLR) test statistic,
\begin{equation} \label{eq:qlr_test_stat}
\hat{T}_N(\theta) = \min \limits_{t \leq \mathbf{0}}\left(\sqrt{N}(\widehat{W}_{\text{test}} - \theta \bm{1}) - t))'\widehat{\Sigma}^{-1}(\sqrt{N}(\widehat{W}_{\text{test}} - \theta \bm{1}) - t)\right)
\end{equation}
also considered in \citet{andrews2010inference}. The least-favorable critical value,
\begin{equation} \label{eq:qlr_crit_val}
\hat{c}_{\alpha, \text{QLR}}^{LF} = \hat{Q}_{1-\alpha} \left( \min \limits_{t \leq \mathbf{0}}\left(\sqrt{N}(\widehat{W}_{\text{test}} - W_{\text{test}}) - t))'\widehat{\Sigma}^{-1}(\sqrt{N}(\widehat{W}_{\text{test}} - W_{\text{test}}) - t)\right) \right),
\end{equation}
can be estimated using bootstrap or Gaussian approximation and requires solving one convex program per simulation. Our final Proposition summarizes this discussion. 

\begin{proposition}[$LCB$ by test inversion with QLR]
\label{prop:qlr}
The LCB obtained by inverting a test in \eqref{eq:test} with the QLR test statistic \eqref{eq:qlr_test_stat} and least-favorable critical value \eqref{eq:qlr_crit_val}  takes the form
\begin{equation} 
\widehat{LCB}_{\text{\normalfont mix}} = \max_{\lambda\in\Lambda} \left\{ \lambda^{\prime} \widehat{W}_{\text{\normalfont test}} - (\hat{c}_{\alpha, \text{QLR}}^{LF})^{1/2} \frac{\sqrt{\lambda'\,\widehat\Sigma\,\lambda}}{\sqrt{N}} \right\}.
\end{equation}
\end{proposition}

In practice, \(\widehat{LCB}_{\max}^{LF}\) in \eqref{eq:lcb_max_lf} (and its GMS version \eqref{eq:lcb_max_gms}) are computationally simpler and employ a less conservative critical value than \(\widehat{LCB}_{\text{mix}}\). However, \(\widehat{LCB}_{\text{mix}}\)  involves searching over all convex mixtures of test policies  which creates more scope to trade off mean welfare against precision. As discussed in Example 4.1 in \citet{canay2017practical}, both tests are admissible, so the corresponding LCBs cannot generally be ranked. Depending on the underlying DGP, either of the LCBs may be tighter. 

\section{Empirical Application}
\label{sec:jc}

To illustrate the welfare-precision trade-off in practice and showcase the proposed procedures, we revisit the National Job Training Partnership Act (JTPA) study, considered in 
\citet{HeckmanIchimuraTodd} and \citet{AbadieAngrist} and recently revisited in the context of policy learning by
\citet{KitagawaTetenov}, \citet{MbakopTabord}, and \citet{AtheyWager2}, among others.  A detailed description of the study is available in \citet{Bloom1997}.

 The study randomized whether applicants would be eligible to receive job training and related services for a period of eighteen months. The treatment $D$ is the indicator of program eligibility. The outcome $Y$ is the applicant's cumulative earnings thirty months after assignment. Two baseline covariates $X=(PreEarn, Educ)$ include pre-program earnings (in USD) and years of education. By design, unconditional independence holds,
$$
(Y(1), Y(0), X) \perp D, 
$$
so the first-best welfare and the corresponding welfare gain are identified in each of the models $(D,Y)$, $(PreEarn, D,Y)$, $(Educ, D,Y)$ and  $(X,D,Y)$. This fact allows us to compare the estimated optimal welfare gains and corresponding LCBs across the models and highlight connections with our theoretical results.

Table \ref{tab:jtpa} presents the estimates and LCBs for the welfare gain based on first-best policy rules in three different policy classes: no covariates (Row 1), only $PreEarn$ (Rows 2--3), and both covariates $X$ (Row 4). The welfare gain from treating everyone (Row 1) corresponds to the Average Treatment Effect. It provides a robust lower bound for the optimal welfare gain based on more complex policy classes, so we use it as a reference point. In Rows 2--3, given that $PreEarn$ is continuously distributed and the margin assumption is plausible, we adopt the cross-fitted efficient-score estimator. We consider estimating the CATE function of $PreEarn$ via series regression (Row 2) and random forest (Row 3). To estimate the propensity score, we bin $PreEarn$ into five cells of similar size and use cell-specific averages as an input into the regression adjustment estimator of the form \eqref{eq:mdx}.

First, in the full model $(X, D, Y)$, we find that the margin assumption likely fails. For eleven out of twelve education groups, the CATE is not significant at the $5\%$ level, so ties among the first-best treatment rules based on $Educ$ are very likely. Although the continuous covariate $PreEarn$ may alleviate the concern, violation of margin assumption can still be detected based on the heuristic in Remark \ref{rm:vanish}. The estimated welfare gap (Row 4, Column 4) is negative yet 23\% of individuals would be treated differently than under the optimal policy (Row 4, Column 1), so the inequality \eqref{eq:marginmain2} is violated in-sample. To ensure validity of the reported LCB, we do not cross-fit. Using only two-thirds of the sample to compute the point estimate and its $95\%$ LCB incurs substantial efficiency loss, resulting in the lowest LCB in the Table. 

Second, in the model $(PreEarn, D,Y)$, we do not detect sufficient heterogeneity to warrant personalized treatment assignment. Comparing Rows 2--3 with Row 1 yields welfare gaps of $-246$ and $-220$, relative to treating everyone. Since the estimated sign is negative,\footnote{ The estimated sign is negative due to the use of the efficient/doubly robust estimators \eqref{eq:dr}, which are not necessarily ordered in-sample.} the true gap is likely of the order sampling error. Moreover, the LCB in Row 1 exceeds those in Rows 2--3 by $40\%$ and $28\%$, respectively. We attribute these findings to potential biases in the first-stage estimators of regression functions and/or lack of heterogeneity in CATE function of $PreEarn$.

Next, we implement the LCBs proposed in Section \ref{sec:moment}, choosing the test policies based on education level. We expect the treatment effects to be non-increasing in education level, with possible jumps at graduation years, $Educ=12$ and $Educ=16$. Thus, we limit the focus on cutoff policies of the form $\{Educ \leq C\}$. In particular, the policy $\{Educ \leq 11\}$ corresponds to treating only those who did not graduate from high-school (37.3\% of the sample); $\{Educ \leq 12\}$ adds those who graduated from high-school but did not attend college (80.0\% of the sample); $\{Educ \leq 15\}$ adds those who attended but did not graduate from college (95.9\% of the sample); $\{Educ \leq 16\}$ adds college graduates (98.7\% of the sample); and $\{Educ \leq 18\}$ corresponds to treating everyone. 

Table \ref{tab:momineq} presents LCBs obtained with the maximum test statistic and different test sets $\mathcal{G}_{\text{test}}$ determined by the cutoffs. In Row 1, the cutoff set corresponds to those who attended but did not graduate from high school and college, as well as everyone in the sample; and Row 2 includes all possible cutoffs. The first-best policy in both classes is $\{Educ \leq 15\}$ with the estimated welfare gain of 1440.25 USD, which exceeds all point estimates in Table \ref{tab:jtpa}. For the first test class, the least-favorable and GMS confidence bands coincide and exceed all of the LCBs in Table \ref{tab:jtpa}. The second test class contains policies that are far from optimal and thus provide loose lower bounds. As a result, the least-favorable test is conservative, while GMS leads to a tighter LCB.

\begin{table}

   \captionof{table}{ Welfare Gain Per Capita: Estimates and Lower Confidence Bands  }
 \label{tab:jtpa}
        \centering 
\begin{tabularx}{\textwidth}{@{}cccccc@{}}
\toprule
Treatment Rule  &  \; \makecell{ Treated \\ Share  } \;  &   \makecell{Welfare Gain \\ (s.e.) } & \makecell{$95\%$ LCB } & \makecell{Welfare Gap \\(USD)} &   \makecell{ Relative LCB \\ Gap $(\%)$ } \\ \midrule 
 Treat Everyone  &  1.00    & \makecell{ 1289.66 \\[1mm] (347.82) } & 717.52   &  -- & -- \\[5mm]
\makecell{Series Regression \\[1mm] $\sum_{j=1}^4 (PreEarn)^j$  }   &   0.992   &  \makecell{1043.22 \\[1mm] (394.67)  } &  394.03  &  -246.44   & $45\%$ \\[5mm]
    
\makecell{Random Forest \\[1mm] $(PreEarn)$ }  &  0.92  & \makecell{ 1069.50  \\[1mm] (335.24) }  & 518.06   & -220.16  & $28\%$  \\[5mm]
            
\makecell{ Random Forest \\[1mm]  $(PreEarn+Educ)$ } &  0.77 & \makecell{ 996.43 \\[1mm] (393.99) }& 348.31 & -293.23  & $51\%$   \\[4mm]
\bottomrule
\end{tabularx}

\caption*{
\textit{Notes:} The outcome variable is 30-Month Post-Program Cumulative Earnings in USD.  Welfare gain is defined in \eqref{eq:welfaregain}.  Row 1: Average Treatment Effect;  Rows 2--3: Welfare gain based on the policy $G = \bm{1}\{ CATE (PreEarn) \geq 0 \}$, where CATE is estimated via series regression or random forest.  Row 4: Sample-split welfare gain based on a plug-in treatment rule estimated via random forest. The $95\%$ LCBs are given by $W^{\text{gain}}-1.645 s.e.(W^{\text{gain}})$. The Welfare Gap is $W^{\text{gain}}_j-ATE$, where $ATE = 1289.66$ (Row 1) and $W^{gain}_j$ are in Rows $j = 2, 3, 4$.  The relative LCB  gap is defined as $100 (1-{LCB}_j/{LCB}_{ATE}) \%$, where ${LCB}_{ATE}=717.52$ (Row 1) and ${LCB}_j$ is in rows $j \in \{2,3,4\}$. The sample  $(N=9,223)$ is the same as in \citet{KitagawaTetenov}.  See text for further details. }
\end{table}

\newcolumntype{Y}{>{\centering\arraybackslash}X}
\begin{table}
   \captionof{table}{ $95\%$ LCB for Welfare Gain   }
 \label{tab:momineq}
\centering 
\begin{tabularx}{\textwidth}{@{}YYY@{}}
    \toprule
   Cutoffs for Test Policies  &  Least-Favorable &  GMS   \\
    \midrule  
 Cutoff $\in \{11, 15, 18\}$ &  783.28 &   783.28 \\[4mm]
Cutoff $\in \{7, 8, \dots,  18\}$ &   649.53 &  724.26 \\
    \bottomrule
\end{tabularx}

\caption*{
\textit{Notes:} The table reports LCBs based on two different test policy classes of the form $\mathcal{G}_{\text{test}} = \{ Educ \leq  C: C \in \mathcal{C}\}$ with a set of cutoffs $\mathcal{C}$ listed above. The policy $\{Educ \leq 18\}$ corresponds to treating everyone. The Generalized Moment Selection (GMS) procedure is from \citet{andrews2010inference}. The critical values $\widehat {c}_{1-\alpha}$ are based on a Gaussian approximation with $10^5$ simulation draws. The sample  $(N=9,223)$ is the same as in \citet{KitagawaTetenov}. See text for further details. }

\end{table}

\section{Conclusion}
\label{sec:concl}
In this paper, we addressed the question of reporting a Lower Confidence Band on the optimal welfare in a policy learning problem. First, we documented the trade-off between welfare and precision and showed that it can be first-order. Second, we connected the first-order trade-off to the lack of uniformity in the margin condition of \citet{MammenTsybakov, Tsybakov}. Finally, we proposed procedures for reporting Lower Confidence Bands that address the trade-off and remain valid regardless of the margin condition.

\bibliographystyle{chicagoa}
\bibliography{my_new_bibtex}

\begin{thebibliography}{}

\bibitem[\protect\citeauthoryear{Abadie, Angrist, and Imbens}{Abadie
  et~al.}{2002}]{AbadieAngrist}
Abadie, A., J.~Angrist, and G.~Imbens (2002).
\newblock Instrumental variables estimates of the effect of subsidized training
  on the quantiles of trainee earnings.
\newblock {\em Econometrica\/}~{\em 70}, 91--117.


\bibitem[\protect\citeauthoryear{Adjaho and Christensen}{Adjaho and
  Christensen}{2022}]{adjaho2023externally}
Adjaho, C. and T.~Christensen (2022).
\newblock Externally valid policy choice.

\bibitem[\protect\citeauthoryear{Andrews and Soares}{Andrews and
  Soares}{2010}]{andrews2010inference}
Andrews, D.~W. and G.~Soares (2010).
\newblock Inference for parameters defined by moment inequalities using
  generalized moment selection.
\newblock {\em Econometrica\/}~{\em 78\/}(1), 119--157.


\bibitem[\protect\citeauthoryear{Andrews and Chen}{Andrews and
  Chen}{2025}]{andrewschen2025}
Andrews, I. and J.~Chen (2025).
\newblock Certified decisions.

\bibitem[\protect\citeauthoryear{Andrews, Kitagawa, and McCloskey}{Andrews
  et~al.}{2024}]{andrews2024inference}
Andrews, I., T.~Kitagawa, and A.~McCloskey (2024).
\newblock Inference on winners.
\newblock {\em The Quarterly Journal of Economics\/}~{\em 139\/}(1), 305--358.


\bibitem[\protect\citeauthoryear{Armstrong}{Armstrong}{2014}]{Armstrong2014}
Armstrong, T.~B. (2014).
\newblock Weighted ks statistics for inference on conditional moment
  inequalities.
\newblock {\em Journal of Econometrics\/}~{\em 181\/}(2), 92--116.


\bibitem[\protect\citeauthoryear{Armstrong}{Armstrong}{2015}]{Armstrong2015}
Armstrong, T.~B. (2015).
\newblock Adaptive testing on a regression function at a point.
\newblock {\em The Annals of Statistics\/}~{\em 43\/}(5), 2086--2101.


\bibitem[\protect\citeauthoryear{Armstrong and Shen}{Armstrong and
  Shen}{2023}]{ArmstrongShen2023}
Armstrong, T.~B. and S.~Shen (2023).
\newblock Inference on optimal treatment assignment.
\newblock {\em The Japanese Economic Review\/}~{\em 74\/}(4), 471--500.


\bibitem[\protect\citeauthoryear{Athey and Wager}{Athey and
  Wager}{2021}]{AtheyWager2}
Athey, S. and S.~Wager (2021, January).
\newblock Policy learning with observational data.
\newblock {\em Econometrica\/}~{\em 89}, 133--161.


\bibitem[\protect\citeauthoryear{Bai, Santos, and Shaikh}{Bai
  et~al.}{2022}]{bai2022two}
Bai, Y., A.~Santos, and A.~M. Shaikh (2022).
\newblock A two-step method for testing many moment inequalities.
\newblock {\em Journal of Business \& Economic Statistics\/}~{\em 40\/}(3),
  1070--1080.


\bibitem[\protect\citeauthoryear{Ben-Michael}{Ben-Michael}{2025}]{benmichael2025}
Ben-Michael, E. (2025).
\newblock Partial identification via conditional linear programs: estimation
  and policy learning.

\bibitem[\protect\citeauthoryear{Ben-Michael, Greiner, Imai, and
  Jiang}{Ben-Michael et~al.}{2021}]{benmichael2021safepolicy}
Ben-Michael, E., D.~J. Greiner, K.~Imai, and Z.~Jiang (2021).
\newblock Safe policy learning through extrapolation: Application to pre-trial
  risk assessment.

\bibitem[\protect\citeauthoryear{Berger}{Berger}{1997}]{Berger1997}
Berger, R.~L. (1997).
\newblock {\em Likelihood Ratio Tests and Intersection-Union Tests}, pp.\
  225--237.
\newblock Boston, MA: Birkhauser Boston.

\bibitem[\protect\citeauthoryear{Bhattacharya and Dupas}{Bhattacharya and
  Dupas}{2012}]{BhatacharyaDupas2012}
Bhattacharya, D. and P.~Dupas (2012).
\newblock Inferring welfare maximizing treatment assignment under budget
  constraints.
\newblock {\em Journal of Econometrics\/}~{\em 167\/}(1), 168--196.


\bibitem[\protect\citeauthoryear{Bloom, Orr, Bell, Cave, Doolittle, Lin, and
  Bos}{Bloom et~al.}{1997}]{Bloom1997}
Bloom, H.~S., L.~L. Orr, S.~H. Bell, G.~Cave, F.~Doolittle, W.~Lin, and J.~M.
  Bos (1997).
\newblock The benefits and costs of jtpa title ii-a programs: Key findings from
  the national job training partnership act study.
\newblock {\em The Journal of Human Resources\/}~{\em 32\/}(3), 549--576.


\bibitem[\protect\citeauthoryear{Canay and Shaikh}{Canay and
  Shaikh}{2017}]{canay2017practical}
Canay, I.~A. and A.~M. Shaikh (2017).
\newblock Practical and theoretical advances in inference for partially
  identified models.
\newblock {\em Advances in Economics and Econometrics\/}~{\em 2}, 271--306.


\bibitem[\protect\citeauthoryear{Chamberlain}{Chamberlain}{2011}]{Chamberlain2011}
Chamberlain, G. (2011, 09).
\newblock {1011 Bayesian Aspects of Treatment Choice}.
\newblock In {\em {The Oxford Handbook of Bayesian Econometrics}}. Oxford
  University Press.

\bibitem[\protect\citeauthoryear{{Chandrasekhar}, {Chernozhukov}, {Molinari},
  and {Schrimpf}}{{Chandrasekhar} et~al.}{2012}]{CCMS}
{Chandrasekhar}, A., V.~{Chernozhukov}, F.~{Molinari}, and P.~{Schrimpf} (2012,
  December).
\newblock {Inference for best linear approximations to set identified
  functions}.
\newblock {\em arXiv e-prints\/}, arXiv:1212.5627.


\bibitem[\protect\citeauthoryear{Chen, Austern, and Syrgkanis}{Chen
  et~al.}{2023}]{chen2023inference}
Chen, Q., M.~Austern, and V.~Syrgkanis (2023).
\newblock Inference on optimal dynamic policies via softmax approximation.

\bibitem[\protect\citeauthoryear{Chernozhukov, Hong, and Tamer}{Chernozhukov
  et~al.}{2007}]{CHT}
Chernozhukov, V., H.~Hong, and E.~Tamer (2007, August).
\newblock Estimation and confidence regions for parameter sets in econometric
  models.
\newblock {\em Econometrica\/}~{\em 75}, 1243--1284.


\bibitem[\protect\citeauthoryear{Chernozhukov, Lee, and Rosen}{Chernozhukov
  et~al.}{2013}]{CLR}
Chernozhukov, V., S.~Lee, and A.~Rosen (2013).
\newblock Intersection bounds: Estimation and inference.
\newblock {\em Econometrica\/}~{\em 81}, 667--737.


\bibitem[\protect\citeauthoryear{Chernozhukov, Lee, Rosen, and
  Sun}{Chernozhukov et~al.}{2025}]{chernozhukov2025polece}
Chernozhukov, V., S.~Lee, A.~M. Rosen, and L.~Sun (2025).
\newblock Policy learning with confidence.

\bibitem[\protect\citeauthoryear{Chernozhukov, Newey, and Santos}{Chernozhukov
  et~al.}{2015}]{CherNeweySantos}
Chernozhukov, V., W.~K. Newey, and A.~Santos (2015).
\newblock Constrained conditional moment restriction models.

\bibitem[\protect\citeauthoryear{Cox}{Cox}{2024}]{cox2024simple}
Cox, G.~F. (2024).
\newblock A simple and adaptive confidence interval when nuisance parameters
  satisfy an inequality.
\newblock {\em arXiv preprint arXiv:2409.09962\/}.


\bibitem[\protect\citeauthoryear{Cui and Han}{Cui and
  Han}{2024}]{cui2024policy}
Cui, Y. and S.~Han (2024).
\newblock Policy learning with distributional welfare.

\bibitem[\protect\citeauthoryear{Dehejia}{Dehejia}{2005}]{DEHEJIA2005141}
Dehejia, R.~H. (2005).
\newblock Program evaluation as a decision problem.
\newblock {\em Journal of Econometrics\/}~{\em 125\/}(1), 141--173.
\newblock Experimental and non-experimental evaluation of economic policy and
  models.


\bibitem[\protect\citeauthoryear{Dumbgen and Spokoiny}{Dumbgen and
  Spokoiny}{2001}]{DumbgenSpokoiny2001}
Dumbgen, L. and Spokoiny (2001).
\newblock Multiscale testing of qualitative hypotheses.
\newblock {\em The Annals of Statistics\/}~{\em 29\/}(1), 124--152.


\bibitem[\protect\citeauthoryear{Hahn}{Hahn}{1998}]{Hahn98}
Hahn, J. (1998, March).
\newblock On the role of the propensity score in efficient semiparametric
  estimation of average treatment effects.
\newblock {\em Econometrica\/}~{\em 66\/}(2), 315--331.


\bibitem[\protect\citeauthoryear{Harter}{Harter}{1964}]{LeonHarter}
Harter, H.~L. (1964).
\newblock Criteria for best substitute interval estimators, with an application
  to the normal distribution.
\newblock {\em Journal of the American Statistical Association\/}~{\em
  59\/}(308), 1133--1140.


\bibitem[\protect\citeauthoryear{Heckman, Ichimura, and Todd}{Heckman
  et~al.}{1997}]{HeckmanIchimuraTodd}
Heckman, J., H.~Ichimura, and P.~Todd (1997).
\newblock Matching as an econometric evaluation estimator:evidence from
  evaluating a job training programme.
\newblock {\em Review of Economic Studies\/}~{\em 64}, 605--654.


\bibitem[\protect\citeauthoryear{Hirano and Porter}{Hirano and
  Porter}{2009}]{HiranoPorter2009}
Hirano, K. and J.~Porter (2009).
\newblock Asymptotics for statistical treatment rules.
\newblock {\em Econometrica\/}~{\em 77}, 1683--1701.


\bibitem[\protect\citeauthoryear{Hirano and Porter}{Hirano and
  Porter}{2012}]{HiranoPorter2012}
Hirano, K. and J.~Porter (2012).
\newblock Impossibility results for nondifferentiable functionals.
\newblock {\em Econometrica\/}~{\em 80}, 1769--1790.


\bibitem[\protect\citeauthoryear{Kallus}{Kallus}{2022a}]{kallus2022treatment}
Kallus, N. (2022a).
\newblock Treatment effect risk: Bounds and inference.

\bibitem[\protect\citeauthoryear{Kallus}{Kallus}{2022b}]{kallus2022whats}
Kallus, N. (2022b).
\newblock What's the harm? sharp bounds on the fraction negatively affected by
  treatment.

\bibitem[\protect\citeauthoryear{Kallus, Mao, and Zhou}{Kallus
  et~al.}{2020}]{kallus2020assessing}
Kallus, N., X.~Mao, and A.~Zhou (2020).
\newblock Assessing algorithmic fairness with unobserved protected class using
  data combination.

\bibitem[\protect\citeauthoryear{Kasy}{Kasy}{2016}]{kasy2016partial}
Kasy, M. (2016).
\newblock Partial identification, distributional preferences, and the welfare
  ranking of policies.
\newblock {\em Review of Economics and Statistics\/}~{\em 98\/}(1), 111--131.


\bibitem[\protect\citeauthoryear{Ketz and McCloskey}{Ketz and
  McCloskey}{2025}]{ketz2025short}
Ketz, P. and A.~McCloskey (2025).
\newblock Short and simple confidence intervals when the directions of some
  effects are known.
\newblock {\em Review of Economics and Statistics\/}~{\em 107\/}(3), 820--834.


\bibitem[\protect\citeauthoryear{Kitagawa, Lee, and Qiu}{Kitagawa
  et~al.}{2022}]{KitagawaLeeChen2022}
Kitagawa, T., S.~Lee, and C.~Qiu (2022).
\newblock Treatment choice with nonlinear regret.
\newblock Working Paper, available at \url{https://arxiv.org/abs/2205.08586}.

\bibitem[\protect\citeauthoryear{Kitagawa and Tetenov}{Kitagawa and
  Tetenov}{2018a}]{KitagawaTetenovSupp}
Kitagawa, T. and A.~Tetenov (2018a).
\newblock Supplement to: Who should be treated? empirical welfare maximization
  methods for treatment choice.
\newblock {\em Econometrica Supplemental Material\/}.


\bibitem[\protect\citeauthoryear{Kitagawa and Tetenov}{Kitagawa and
  Tetenov}{2018b}]{KitagawaTetenov}
Kitagawa, T. and A.~Tetenov (2018b).
\newblock Who should be treated? empirical welfare maximization methods for
  treatment choice.
\newblock {\em Econometrica\/}~{\em 86}, 591--616.


\bibitem[\protect\citeauthoryear{Kitagawa and Tetenov}{Kitagawa and
  Tetenov}{2021}]{kitagawa2021equality}
Kitagawa, T. and A.~Tetenov (2021).
\newblock Equality-minded treatment choice.
\newblock {\em Journal of Business \& Economic Statistics\/}~{\em 39\/}(2),
  561--574.


\bibitem[\protect\citeauthoryear{Lehmann}{Lehmann}{1959}]{Lehmann1959}
Lehmann, E.~L. (1959).
\newblock {\em Testing Statistical Hypotheses}.
\newblock New York: John Wiley \& Sons.


\bibitem[\protect\citeauthoryear{Lehmann and Romano}{Lehmann and
  Romano}{2005}]{LehmannRomano}
Lehmann, E.~L. and J.~P. Romano (2005).
\newblock {\em Testing Statistical Hypotheses\/} (3 ed.).
\newblock Springer Texts in Statistics. New York: Springer.


\bibitem[\protect\citeauthoryear{Levis, Bonvini, Zeng, Keele, and
  Kennedy}{Levis et~al.}{2023}]{levis2023covariateassisted}
Levis, A.~W., M.~Bonvini, Z.~Zeng, L.~Keele, and E.~H. Kennedy (2023).
\newblock Covariate-assisted bounds on causal effects with instrumental
  variables.

\bibitem[\protect\citeauthoryear{Liu and Molinari}{Liu and
  Molinari}{2024}]{liu2025}
Liu, Y. and F.~Molinari (2024).
\newblock Inference for an algorithmic fairness-accuracy frontier.

\bibitem[\protect\citeauthoryear{Luedtke and van~der Laan}{Luedtke and van~der
  Laan}{2016}]{LuedtkeLaan}
Luedtke, A. and M.~van~der Laan (2016).
\newblock Statistical inference for the mean outcome under a possibly
  non-unique optimal treatment strategy.
\newblock {\em Annals of Statistics\/}~{\em 44\/}(2), 713--742.


\bibitem[\protect\citeauthoryear{Mammen and Tsybakov}{Mammen and
  Tsybakov}{1999}]{MammenTsybakov}
Mammen, E. and A.~B. Tsybakov (1999).
\newblock {Smooth discrimination analysis}.
\newblock {\em The Annals of Statistics\/}~{\em 27\/}(6), 1808 -- 1829.


\bibitem[\protect\citeauthoryear{Manski}{Manski}{2004}]{Manski2004}
Manski, C.~F. (2004).
\newblock Statistical treatment rules for heterogeneous populations.
\newblock {\em Econometrica\/}~{\em 72}, 1221--1246.


\bibitem[\protect\citeauthoryear{Mbakop and Tabord-Meehan}{Mbakop and
  Tabord-Meehan}{2021}]{MbakopTabord}
Mbakop, E. and M.~Tabord-Meehan (2021, March).
\newblock Model selection for treatment choice: Penalized welfare maximization.
\newblock {\em Econometrica\/}~{\em 89}, 825--848.


\bibitem[\protect\citeauthoryear{Moon}{Moon}{2025}]{Moon:25}
Moon, S. (2025).
\newblock Optimal policy choices under uncertainty.
\newblock Working Paper, available at \url{https://arxiv.org/abs/2503.03910}.

\bibitem[\protect\citeauthoryear{Murphy}{Murphy}{2003}]{Murphy2003}
Murphy, S.~A. (2003).
\newblock Optimal dynamic treatment regimes.
\newblock {\em Journal of the Royal Statistical Society: Series B (Statistical
  Methodology)\/}~{\em 65\/}(2), 331--355.


\bibitem[\protect\citeauthoryear{Qian and Murphy}{Qian and
  Murphy}{2011}]{QianMurphy}
Qian, M. and S.~A. Murphy (2011).
\newblock {Performance guarantees for individualized treatment rules}.
\newblock {\em The Annals of Statistics\/}~{\em 39\/}(2), 1180 -- 1210.


\bibitem[\protect\citeauthoryear{Rai}{Rai}{2019}]{Rai2019}
Rai, Y. (2019).
\newblock Statistical inference for treatment assignment policies.
\newblock Working Paper, available at
  \url{https://sites.google.com/site/yoshiyasurai/home}.

\bibitem[\protect\citeauthoryear{Robins and Rotnitzky}{Robins and
  Rotnitzky}{1995}]{Robins}
Robins, J. and A.~Rotnitzky (1995).
\newblock Semiparametric efficiency in multivariate regression models with
  missing data.
\newblock {\em Journal of American Statistical Association\/}~{\em 90\/}(429),
  122--129.


\bibitem[\protect\citeauthoryear{Robins}{Robins}{2004}]{Robins2004}
Robins, J.~M. (2004).
\newblock Optimal structural nested models for optimal sequential decisions.
\newblock In D.~Y. Lin and P.~J. Heagerty (Eds.), {\em Proceedings of the
  Second Seattle Symposium in Biostatistics}, Volume 179 of {\em Lecture Notes
  in Statistics}, pp.\  189--326. New York, NY: Springer.

\bibitem[\protect\citeauthoryear{Romano, Shaikh, and Wolf}{Romano
  et~al.}{2014}]{romano2014practical}
Romano, J.~P., A.~M. Shaikh, and M.~Wolf (2014).
\newblock A practical two-step method for testing moment inequalities.
\newblock {\em Econometrica\/}~{\em 82\/}(5), 1979--2002.


\bibitem[\protect\citeauthoryear{Sasaki and Ura}{Sasaki and
  Ura}{2024}]{SasakiUra}
Sasaki, Y. and T.~Ura (2024).
\newblock Welfare analysis with marginal treatment effects.
\newblock {\em Econometric Theory\/}.


\bibitem[\protect\citeauthoryear{Semenova}{Semenova}{2020}]{SemSupp2}
Semenova, V. (2020).
\newblock Generalized lee bounds.

\bibitem[\protect\citeauthoryear{Semenova}{Semenova}{2023}]{Semenova2024}
Semenova, V. (2023).
\newblock Aggregated intersection bounds and aggregated minimax values.

\bibitem[\protect\citeauthoryear{Shi, Fan, Song, and Lu}{Shi
  et~al.}{2018}]{ShiEtAl2018}
Shi, C., A.~Fan, R.~Song, and W.~Lu (2018).
\newblock High-dimensional a-learning for optimal dynamic treatment regimes.
\newblock {\em The Annals of Statistics\/}~{\em 46\/}(3), 925--957.


\bibitem[\protect\citeauthoryear{Stoye}{Stoye}{2009}]{Stoye}
Stoye, J. (2009).
\newblock Minimax regret treatment choice with finite samples.
\newblock {\em Journal of Econometrics\/}~{\em 151}, 70--81.


\bibitem[\protect\citeauthoryear{Sun}{Sun}{2021}]{Sun}
Sun, L. (2021).
\newblock Empirical welfare maximization with constraints.


\bibitem[\protect\citeauthoryear{Terschuur}{Terschuur}{2025}]{terschuur2025locally}
Terschuur, J. (2025).
\newblock Locally robust policy learning: Inequality, inequality of opportunity
  and intergenerational mobility.
\newblock {\em arXiv preprint arXiv:2502.13868\/}.


\bibitem[\protect\citeauthoryear{Tetenov}{Tetenov}{2012}]{Tetenov2012}
Tetenov, A. (2012).
\newblock Statistical treatment choice based on asymmetric minimax regret
  criteria.
\newblock {\em Econometrica\/}~{\em 166}, 157--165.


\bibitem[\protect\citeauthoryear{Tsybakov}{Tsybakov}{2004}]{Tsybakov}
Tsybakov, A.~B. (2004).
\newblock {Optimal aggregation of classifiers in statistical learning}.
\newblock {\em The Annals of Statistics\/}~{\em 32\/}(1), 135 -- 166.


\bibitem[\protect\citeauthoryear{Whitehouse, Austern, and Syrgkanis}{Whitehouse
  et~al.}{2025}]{whitehouse2025}
Whitehouse, J., M.~Austern, and V.~Syrgkanis (2025).
\newblock Inference on optimal policy values and other irregular functionals
  via smoothing.

\bibitem[\protect\citeauthoryear{Yata}{Yata}{2021}]{Yata2021}
Yata, K. (2021).
\newblock Optimal decision rules under partial identification.
\newblock Working Paper, available at \url{https://arxiv.org/abs/2111.04926}.

\end{thebibliography}

\newpage 

\appendix

\section{Proofs for Sections 3--4} \label{app:proofs}

Section \ref{sec:LvdL} contains auxiliary statements and the proof of  \eqref{eq:mainstatement}.   The proof of \eqref{eq:mse} is given in  Section \ref{sec:proof}. Section \ref{sec:sec2} contains the proof of Proposition \ref{thm:equivalence}.

\subsection{Auxiliary statements}
\label{sec:LvdL}

The first Lemma is Theorem 1 in \citet{LuedtkeLaan}. 
\begin{lemma}[Efficiency Influence Function for the first-best welfare]
\label{lem:LvdL}
Suppose Assumption \ref{ass:overlap} holds and $P(\tau(X)=0)=0$. Then, the first-best welfare  $\E[\max (m(1,X), m(0,X))]$ is pathwise differentiable with efficient influence function 
\begin{align}
\label{eq:efficient}
\psi^{*}(Z) &= \bigg(m(1,X) + \dfrac{D}{\pi(X)} (Y - m(1,X)) \bigg)  \bm{1}\{ \tau(X)>0 \} \\
&+ \bigg(m(0,X) +  \dfrac{1-D}{1-\pi(X)} (Y - m(0,X)) \bigg) \bm{1}\{ \tau(X)<0 \}. \nonumber
\end{align}
\end{lemma}

The second Lemma is a simple corollary of \citet{Hahn98}.
\begin{lemma}[Efficiency bound for $W_G$]
\label{lem:LvdLg}
Suppose Assumption \ref{ass:overlap} holds and let $G$ be a \textit{known} policy. Then, the welfare $W_G$ is pathwise differentiable with efficient influence function 
\begin{align*}
\psi_G(Z) &= \bigg(m(1,X) + \dfrac{D}{\pi(X)} (Y - m(1,X)) \bigg)  \bm{1}\{ X \in G \} \\
&+ \bigg(m(0,X) +  \dfrac{1-D}{1-\pi(X)} (Y - m(0,X)) \bigg) \bm{1}\{ X \in G^c \}. \nonumber
\end{align*}
The corresponding variance  is
\begin{multline} \label{eq:var_bound}
	\sigma^2_G = \Var(m(1, X)\bm{1}(X \in G) + m(0, X) \bm{1}(X \in G^c))\\ + \E\left[\frac{\sigma^2(1, X)}{\pi(X)}\bm{1}(X \in G) + \frac{\sigma^2(0, X)}{1-\pi(X)}\bm{1}(X\notin G)\right],
\end{multline}
where $\sigma^2(d, x) = \Var(Y(d) \,|\,X = x) =\Var(Y \,|\,D = d, D = x)$.
\end{lemma}
\begin{proof}
The parameter $W_G = \E [ Y(1) \bm{1}\{ X \in G \} + Y(0) \bm{1}\{ X \in G^c  \}]$ is a sum of two potential outcomes weighted by known functions of $X$, namely, $\bm{1}\{ X \in G \}$ and $\bm{1}\{ X \in G^c \}$.  The form of the efficient influence function $\psi_G(Z) - W_G$ follows immediately from \citet{Hahn98}, Theorem 1.
	 By the Law of Total Variance,
	\[
	\Var(\psi_G(Z)) = \Var(\E[\psi_G(Z) \,|\,X]) + \E[\Var(\psi_G(Z)\,|\,X)].
	\]
	By the Law of Iterated Expectations,
	\[
	\E[\psi_G(Z)\,|\,X] = m(1, X) \bm{1}(X \in G) + m(0, X) \bm{1}(X \in G^c),
	\]
	and $\E[\psi_G^2(Z) \,|\, X]$ takes the form
	\[
	\begin{array}{cl}
		\E[\psi_G^2(Z) \,|\, X] &= m(1, X)^2 \bm{1}(X \in G) + m(0, X)^2 \bm{1}(X \in G^c)\\[3mm]
		&+ \dfrac{\sigma^2(1, X)}{\pi(X)}\bm{1}(X \in G) + \dfrac{\sigma^2(0, X)}{1-\pi(X)} \bm{1}(X \in G^c).
	\end{array}
	\]
	As a result,
	\[
	\Var(\psi_G(Z) \,|\, X) = \dfrac{\sigma^2(1, X)}{\pi(X)}\bm{1}(X \in G) + \dfrac{\sigma^2(0, X)}{1-\pi(X)} \bm{1}(X \in G^c),
	\]
	and the stated formula follows.
\end{proof}

Note that plugging the unconstrained first best policy, $G^*=\{x \in \mathcal{X}: \tau(x) \geqslant 0\}$ into $\psi_{G}(Z)$, that is $\psi_{G^*}(Z) = \psi^{*}(Z)$. Thus, the efficiency bound is the same as if the first-best policy $G^{*}$ was known. 

The next Lemma states a uniform lower bound on $\Var(\psi_G)$, which is useful in the sequel.

\begin{lemma}[A lower bound on variance]
\label{lem:dr-cov-lower}
Suppose Assumption~\ref{ass:overlap}(1) holds and, for each $d\in\{0,1\}$,
\[
\operatorname*{ess\,inf}_{x}\Var\left(Y(d)\,|\,X=x\right)\ \ge\ \underline{\sigma}^2>0.
\]
Then, for any policy $G\subseteq\mathcal X$,
\[
\Var\!\big(\psi_G(Z)\big)\ \ge\ \underline{\sigma}^2.
\]
\end{lemma}

\begin{proof}
Follows immediately from \eqref{eq:var_bound} and the fact that $\pi(X) \in (0, 1)$. 
\end{proof}

Lemma \ref{lem:mylem} shows that for the DGP in Section \ref{sec:dgp}, the welfare gap is proportional to $\epsilon = o(1)$ while the corresponding efficiency bounds remain strictly separated. As a result, it gives the proof for the second part of Proposition 1.

\begin{lemma}[Separated efficiency bounds]
\label{lem:mylem}
The following calculations hold:

\begin{enumerate}
	\item $W_{\mathcal{X}} = 1/2 - \epsilon p$,\;\; $\sigma^2_{\mathcal{X}} =  \dfrac{p}{ \pi (1)}  + \dfrac{(1-p) }{ \pi (0)} + \epsilon^2 (1-p)p.$
\item $W_{G^{*}} = 1/2,$ \;\;$\sigma^2_{G^{*}} = \dfrac{10 p}{ 1-\pi (1)}
 + \dfrac{1-p}{ \pi (0)}$.

\item $\sigma^2_{G^{*}}-\sigma^2_{\mathcal{X}} > 8p$ and $\sigma_{G^{*}}-\sigma_{\mathcal{X}} > p.$

\item $\Delta_{\mathcal{X}} > z_{1-\alpha} p /\sqrt{N}$, for each $N > z_{1-\alpha}^2$
\end{enumerate}
\end{lemma}

\begin{proof}[Proof]

\textbf{Part 1.}  The value of $W_{\mathcal{X}}$ is computed in the main text. The efficiency bound for $W_{\mathcal{X}}$ in \eqref{eq:var_bound} consists of two summands.
The first summand is
$$
\E \bigg[(m(1,X) - W_{\mathcal{X}})^2 \bigg]= \epsilon^2 (1-p)^2 p + \epsilon^2 p^2 (1-p) = \epsilon^2 p (1-p),
$$
and the second is
$$
\E \bigg[ \dfrac{  \Var (Y \mid D=1, X) }{\pi(X)}  \bigg] = \dfrac{1}{\pi(1)} p + \dfrac{1}{\pi(0)} (1-p).
$$
 Adding them up gives $\sigma^2_{\mathcal{X}}$. 

\noindent 
\textbf{Part 2.} The value of $W_{G^*}$ is computed in the main text. The efficiency bound of $W_{G^{*}}$ in \eqref{eq:var_bound} consists of two summands. The first summand is
 \[
 \Var (\max (m(1,X), m(0,X))) =\Var(1/2) =0,
 \]
 and the second one is 
\begin{align*}
&\E \bigg[\dfrac{ X \Var (Y \mid D=0, X=1) }{1-\pi(1)} + \dfrac{ (1-X) \Var (Y \mid D=1, X=0) }{\pi(0)} \bigg] = \dfrac{10 p}{ 1-\pi (1)}
 + \dfrac{1-p}{ \pi (0)}.
\end{align*}
Adding them up yields $\sigma^2_{G^{*}}$. 

\noindent 
\textbf{Part 3.} Recall that $\pi(0), \pi(1), p \in (1/4, 3/4)$. From Parts (1) and (2) it follows that

\[
\sigma^2_{G^*} - \sigma^2_{\mathcal{X}} = p \frac{11\pi(1) - 1}{\pi(1)(1-\pi(1)} - \epsilon^2(1-p)p \; \overset{(i)}{>} \; p \frac{\pi(1)^2 + 10\pi(1) - 1}{\pi(1)(1-\pi(1))} \; \overset{(ii)}{\geq} \; 8p,
\]
where (i) holds by $\epsilon^2(1-p) < 1$ and (ii) is attained at $\pi(1) = 1/4$, which can be verified numerically. 

\noindent 
\textbf{Part 4}. Note that, 
\[
\sigma^2_{G^*} = \frac{10p}{1-\pi(1)} + \frac{1-p}{\pi(0)} \leq 40p + 4(1-p) = 4 + 36p \leq 31.
\] 
Similarly,
\[
\sigma^2_{\mathcal{X}} = \frac{p}{\pi(1)} + \frac{1-p}{\pi(0)} + \epsilon^2(1-p)p \leq 4p + 4(1-p) + \epsilon^2(1-p)p \leq 4 + \frac{1}{4} \varepsilon^2 \leq 5.
\]
Therefore,
\[
\sigma_{G^{*}}-\sigma_{\mathcal{X}} = \dfrac{\sigma^2_{G^{*}}-\sigma^2_{\mathcal{X}}}{\sigma_{G^{*}}+\sigma_{\mathcal{X}}}  \geq \dfrac{8p}{\sqrt{31} + \sqrt{5}} > p.
\]

\noindent 
\textbf{Part 5} Combining the above results, we obtain
\[
\Delta_{\mathcal{X}} = \frac{z_{1-\alpha}}{\sqrt{N}}(\sigma_{G^*} - \sigma_{\mathcal{X}}) - (W_{G^*} - W_{\mathcal{X}}) > \frac{z_{1-\alpha}}{\sqrt{N}}p - \epsilon p \geqslant   \frac{pz_{1-\alpha}}{\sqrt{N}},
\]
for $\epsilon \leq z_{1-\alpha} / \sqrt{N}$. It remains to ensure that $\epsilon^2(1-p) < 1$ which results in a bound on $N$.
 \end{proof}

\subsection{Proof of Proposition \ref{thm:main1}}
\label{sec:proof}

\paragraph{Notation and Preliminaries.} Recall the class of DGPs defined in Section \ref{sec:dgp} and the notation introduced in Section \ref{sec:est_lcb}. Note that $N_{dx} \sim Binom (N, \rho)$,  $\rho = P(D=d, X=x)$, and $\hat{p} = \sum_{i = 1}^N X_i / N \sim Binom(N, p) / N$. For any $\mathcal{Z} \sim Binom (N, \rho)$, for $N \geq 1$,  the following standard properties hold:
  \begin{align*}
\E [  (\mathcal{Z}+ 1)^{-1}   ] & = \dfrac{1}{(N+1) \rho } - \dfrac{1}{(N+1)  \rho } (1- \rho )^{N+1}  \leq N^{-1} \rho^{-1}; \\
 \E [   (\mathcal{Z} + 1)^{-2}   ]   & \leq 2  \rho^{-2}  N^{-2}.
   \end{align*}  
  Moreover, the following Chernoff's bounds hold, with $\mu = \E[\mathcal{Z}]$ and $\delta \in (0, 1)$,
  \begin{equation}
  P(\mathcal{Z} \geq (1+\delta)\mu) \leq \exp\left(-\frac{\delta^2 \mu}{3} \right); \;\;\;\;\; P(\mathcal{Z}\leq (1-\delta)\mu) \leq \exp \left(-\frac{\delta^2\mu}{2} \right)
  \end{equation}
We focus on symmetric DGPs with $p=\pi(1) = \pi(0) = 1/2$, so $\rho =1/4$ for all pairs $(d,x)$. In this case,
\begin{align}
\E [  (N_{dx} + 1)^{-1}   ] & = \dfrac{4}{(N+1) } - \dfrac{4}{(N+1)  } (3/4)^{N+1}  \leq 4/N \label{eq:mean1} \\
 \E [   (N_{dx} + 1)^{-2}   ]   & \leq 32  N^{-2}. \label{eq:mean3}
\end{align} 
For any sequence $C_N \in (0, 1)$, letting $S = \sum_{i = 1}^NX_i \sim Binom(N, 1/2)$ with $\E[S] = N/2$,
\begin{equation} \label{eq:chernoff}
\begin{array}{cl}
	P\left(\left|\hat{p} - 1/2\right| \geq C_N\right) &= P(\hat{p} \geq 1/2 + C_N) + P(\hat{p} \leq 1/2 - C_N)\\[3mm]
	&= P(S \geq (1 + 2C_N)\frac{N}{2}) + P(S \leq (1-2C_N)\frac{N}{2})\\[3mm]
	&\leq \exp \left(-\frac{2}{3}N(C_N)^2 \right) + \exp\left(-N(C_N)^2 \right)\\[3mm]
	&\leq 2\exp \left(-\frac{2}{3}N(C_N)^2 \right).
\end{array}
\end{equation}

\paragraph{Structure of the proof.}   Lemma \ref{lem:bias} bounds the approximation error of expected conditional variance. Lemma \ref{lem:oracle} establishes a lower bound for $MSE (\widehat W_{G^{*}})$. Lemma \ref{lem:m0mse} establishes an upper bound for $MSE (\widehat{W}_{\mathcal{X}})$.  Lemma \ref{lem:mylem21}  completes the proof.

 \begin{lemma}
   \label{lem:bias}
 For $N \geq 100$ and $C_N = \sqrt{2.25 \ln N/N}$, the following bounds hold for any $d,x \in \{1,0\}$
\begin{align}
 (1-10 C_N) < \E [\sigma^{-2}(d,1)  N \cdot \Var ( \widehat m_{d1} \mid \mathbf{X}, \mathbf{D})\cdot    \widehat{p}^2] <  (1+ 10 C_N). \label{eq:varbound} \\
(1-10 C_N) <   \E [\sigma^{-2}(d,0) N \cdot \Var ( \widehat m_{d0} \mid \mathbf{X}, \mathbf{D})   \cdot (1-\widehat{p})^2] < (1+ 10 C_N). \label{eq:varbound2} 
\end{align}
\end{lemma}

\begin{proof}

\textbf{Step 1 (Notation).} Denote the expression inside the expectation of \eqref{eq:varbound} by
$$
\Xi_d = \sigma^{-2}(d,1)  N \Var ( \widehat m_{d1} \mid \mathbf{X}, \mathbf{D})   \widehat{p}^2 = N N_{d1}(N_{d1}  + 1)^{-2} \widehat{p}^2,
$$
and note that its probability limit  as $N \to \infty$ equals $1$. Denoting 
$$\psi^1_{d} (t)= N t  (N_{d1} + 1)^{-1}; \;\; \quad \psi^2_{d} (t)= N t  (N_{d1} + 1)^{-2},
$$ 
we can decompose the asymptotic error as 
\begin{align}
\label{eq:vardecomp}
\Xi_d -1 &= N N_{d1}(N_{d1}  + 1)^{-2} \widehat{p}^2-1\\[2mm]
&= \psi^1_{d} (\widehat{p}^2) - \psi^2_{d}(\widehat{p}^2) - 1 \nonumber \\[2mm]
&=  \psi^1_{d} (\widehat{p}^2 -p^2) + \psi^1_d (p^2) -  \psi^2_{d}(\widehat{p}^2) - 1 \nonumber \\[2mm]
&=\underbracket{\psi^1_{d} (\widehat{p}^2-p^2)}_{S_2}   + \underbracket{\psi^1_{d} (p^2) -  N/(N+1) }_{S_1} - \underbracket{\psi^2_{d}(\widehat{p}^2)}_{S_3}  -\underbracket{1/(N+1)}_{S_4}. \nonumber 
\end{align}

\noindent 
\textbf{Step 2 (Leading term $S_2$).}  Recall that $p = 1/2$. On the event $\mathcal{M}_N = \{ |\widehat p - 1/2| <  C_N \}$, the error $ |\widehat p^2 -1/4| \leq |\widehat{p} - 1/2||\widehat{p} + 1/2| \leq 1.5C_N$. As a result, 
\begin{align}
  |\E\left[S_2 \bm{1}\{ \mathcal{M}_N  \}\right] | &\leq  \E \left[ \psi^1_{d} (|\widehat p^2 - p^2|)  \bm{1}\{ \mathcal{M}_N  \}\right]  \nonumber \\[2mm]
   &\leq \E [ \psi^1_{d} (1.5C_N)   \bm{1}\{ \mathcal{M}_N  \}] \nonumber  \\[2mm]
   &\leq 1.5C_N \E [ \psi_{d}^1(1)] \nonumber \\[2mm]
   & \leq 6C_N, \label{eq:6cn}
\end{align}
 where the first three lines follow from linearity of $\psi^1_d(\cdot)$ and monotonicity of expectation and the last one follows from \eqref{eq:mean1}.  On the  event $\mathcal{M}_N^c$, we can bound $ |\widehat p^2 -p^2| \leq 1 \text{, a.s.,}$ so that
\begin{align*}
    |\E [S_2 \bm{1}\{ \mathcal{M}^c_N  \}] | &\leq \E [ \psi^1_{d} (|\widehat p^2 - p^2|)  \bm{1}\{ \mathcal{M}^c_N  \}]  \\[2mm]
   & \leq \E [ \psi^1_{d} (1)  \bm{1}\{ \mathcal{M}^c_N  \}] \\[2mm]
   & \leq N  P (\mathcal{M}_N^c), 
   \end{align*}   
   where the last line follows from $N_{d1}\geq 0$ and $(N_{d1} +1)^{-1} \leq 1$, a.s.. Using \eqref{eq:chernoff} and $C_N = \sqrt{ 2.25 \ln N/N}$,
\begin{align}
 \textstyle N P (\mathcal{M}^c_N)  \leq 2 N \exp(-\frac{2}{3}N (C_N)^2 ) = 2 N^{-1/2} \leq C_N, \qquad \forall N \geq 6. \label{eq:7cn}
\end{align}
Adding \eqref{eq:6cn} and \eqref{eq:7cn} gives
$ |\E[S_2]|  \leq 7 C_N. $ 

\medskip  

\noindent \textbf{Step 3 (Terms $S_1, S_3, S_4$).} 
 Note that $S_4 =(N+1)^{-1} \leq C_N$.   Invoking \eqref{eq:mean1} gives
  $$
   |\E [S_1] | = 1/4 |\E [ \psi^1_{d} (1) - 4N/(N+1)] | =  N/(N+1) (3/4)^{N+1} \leq   C_N, \quad \forall N \geq 2.
   $$ 
Invoking \eqref{eq:mean3} gives
 $$
  0 \leq \E [S_3]  =\E [ \psi^2_{d} (\widehat{p}^2) ] \leq \E [ \psi^2_{d} (1) ]  \leq  32   N^{-1} \leq  C_N, \quad \forall N \geq 100.
  $$ 
Combining the bounds gives
$$
|\E [ \Xi_d -1] | \leq \sum_{j=1}^4 |\E [ S_d ]| \leq 10 C_N.
$$

\noindent \textbf{Step 4 (Conclusion).} Steps 1--3 established \eqref{eq:varbound}, which corresponds to $x=1$.  The symmetry of DGPs implies \eqref{eq:varbound2} with $x=0$. 
\end{proof}

  \begin{lemma}
  \label{lem:oracle}
 For $N \geq 100$ and $C_N = \sqrt{2.25 \ln N/N}$, $MSE (\widehat W_{G^{*}})$ is lower bounded as 
 \begin{align}
N \cdot MSE (\widehat W_{G^{*}}) > (\sigma^2(1,0) + \sigma^2(0,1)) (1- 10C_N). \label{eq:oraclemse}
 \end{align}
  \end{lemma}

\begin{proof}[Proof]

\textbf{Step 1.} Let $(\mathbf{X}, \mathbf{D})= (X_i, D_i)_{i=1}^N$ be stacked realizations of $(X_i)_{i=1}^N$ and $(D_i)_{i=1}^N$.  For any $i,j \in \{1,2,\dots, N\}$, we show that
$$
X_i (1-X_j)\Cov(Y_i, Y_j \mid \mathbf{X}, \mathbf{D})=0, \text{ a.s.}
$$  
If the indices are distinct,  $\Cov(Y_i, Y_j \mid \mathbf{X}, \mathbf{D})  =0$ by independence of the samples $i$ and $j$. If the indices coincide, the product $X_i(1-X_j) = X_i (1-X_i) =0$ \text{a.s.} Noting that 
$ \text{Cov} (\widehat m_{d_1 1}, \widehat m_{d_2 0} \mid \mathbf{X}, \mathbf{D})$ consists of $N^2$ summands of the form $X_i (1-X_j)\text{cov} (Y_i, Y_j \mid \mathbf{X}, \mathbf{D})$, we obtain 
   $$
  \Cov (\widehat m_{d_1 1}, \widehat m_{d_2 0} \mid \mathbf{X}, \mathbf{D})  = 0, \qquad \forall d_1, d_2 \in \{1, 0\}.
 $$
Thus,  the variance of each estimator is 
 \begin{align}
  \Var  (\widehat{W}_{G^{*}} \mid \mathbf{X}, \mathbf{D})  &= \Var ( \widehat m_{01} \mid \mathbf{X}, \mathbf{D})   \widehat{p}^2 + \Var ( \widehat m_{10} \mid \mathbf{X}, \mathbf{D})   (1-\widehat{p})^2 \label{eq:varghat} \\
  \Var  (\widehat{W}_{\mathcal{X}} \mid \mathbf{X}, \mathbf{D})  &= \Var ( \widehat m_{11} \mid \mathbf{X}, \mathbf{D})   \widehat{p}^2 + \Var ( \widehat m_{10} \mid \mathbf{X}, \mathbf{D})   (1-\widehat{p})^2. \label{eq:varghat2} 
\end{align}

\noindent \textbf{Step 2.}  Invoking Lemma \ref{lem:bias} with $(d,x)=(0,1)$ and $(d,x)=(1,0)$ gives a lower bound
\begin{align}
 \E [N \cdot\Var ( \widehat m_{01} \mid \mathbf{X}, \mathbf{D})   \widehat{p}^2]&> \sigma^2(0,1) (1-10 C_N)  \label{eq:01dx} \\
 \E [ N \cdot \Var ( \widehat m_{10} \mid \mathbf{X}, \mathbf{D})  (1- \widehat{p})^2]&>  \sigma^2(1,0)(1 - 10 C_N) \label{eq:10dx}
\end{align}
Adding  \eqref{eq:01dx} and \eqref{eq:10dx} gives a lower bound on $\E [ \Var  (\widehat{W}_{G^{*}} \mid \mathbf{X}, \mathbf{D})]$. A lower bound \eqref{eq:oraclemse} on $ MSE (\widehat W_{G^{*}})$ follows.
\end{proof}

 \begin{lemma}
  \label{lem:m0mse}
 For $N \geq 100$ and $C_N = \sqrt{2.25 \ln N/N}$ and $N \epsilon^2 \leq 1$,  $ MSE (\widehat{W}_{\mathcal{X}})$ is upper bounded by 
 \begin{align}
 N \cdot MSE (\widehat{W}_{\mathcal{X}}) &<\sigma^2_{\mathcal{X}}+\dfrac{N\epsilon^2}{2}+ (4+10(\sigma^2(1,1)+\sigma^2(1,0)))C_N.\label{eq:m0mse}
 \end{align}
  \end{lemma}

    \begin{proof}[Proof]
 \textbf{Step 1 (Bias).\;}   The remainder term $ R= W_{\mathcal{X}} - \widehat{W}_{\mathcal{X}}$ takes the form
\begin{align}
R= (1/2-\epsilon) \widehat p (N_{11} + 1)^{-1} +1/2(1-\widehat p ) (N_{10}  + 1)^{-1} \label{eq:remainder}
\end{align}
and is non-negative a.s. for $\epsilon \in (0, 1/2)$.  Furthermore, it is bounded as
\begin{align*}
0 \leq  \E [ R] \overset{(i)}{\leq} 1/2\E [  (N_{11} + 1)^{-1} ] + 1/2\E [  (N_{10} + 1)^{-1} ] \overset{(ii)}{\leq} 4/N,
\end{align*}
where (i) follows from the monotonicity of expectation and $\widehat p \leq 1, \text{ a.s.}, $ and (ii) from the standard property of binomial distribution stated in \eqref{eq:mean1}. Next, note that $\E [ 1/2 - \epsilon \widehat p ] = 1/2 - \epsilon p = W_{\mathcal{X}}$ since $\E [ \widehat p ] = p$.  The bias is bounded from above and below
 \begin{align}
0 \leq |W_{G^{*}} -  \E [ \widehat{W}_{\mathcal{X}}]| &\leq  |W_{G^{*}} -  W_{\mathcal{X}}| + |W_{\mathcal{X}} -  \E [ \widehat{W}_{\mathcal{X}}] | \nonumber  \\
&= |W_{G^{*}} -  W_{\mathcal{X}} |+ |\E [ R ]| \nonumber \\
 &\leq \epsilon/2 + 4 N^{-1}. \label{eq:m0mean} 
\end{align}

\noindent 
\textbf{Step 2  (Variance).}\; We show that variance  is upper bounded by
 \begin{align}
N \cdot \Var (\widehat{W}_{\mathcal{X}}) &<\sigma^2_{\mathcal{X}} + (\sigma^2(1,1)+\sigma^2(1,0))10 C_N +2 C_N   .\label{eq:m0var}
 \end{align}
The variance of the conditional mean is 
\begin{align}
\Var (\E [ \widehat{W}_{\mathcal{X}} \mid \mathbf{X}, \mathbf{D}  ]) &= \Var (R) - 2  \Cov(R, 1/2-\epsilon  \widehat p )+ \dfrac{\epsilon^2}{4N}. \label{eq:condmean}
\end{align}
Invoking \eqref{eq:mean3} bounds the variance of the remainder
$$
N \Var (R)  \leq 2/4 \E [N  (N_{11} + 1)^{-2}] + 2/4 \E [  N(N_{10}  + 1)^{-2} ] \leq 32 N^{-1} \leq C_N, \quad \forall N \geq 100.
$$
Invoking  Cauchy inequality bounds the covariance term
$$
2 N | \Cov(R, 1/2-\epsilon  \widehat p ) |  \leq 2\sqrt{ 32 \epsilon^2/(4N)  } \leq 4 \sqrt{2} N^{-1} \leq C_N, \quad \forall N \geq 8.
$$
Invoking \eqref{eq:varbound} gives 
\begin{align}
\label{eq:condvar}
 \E [ N \cdot \Var (\widehat{W}_{\mathcal{X}} \mid \mathbf{X}, \mathbf{D}  ) ]  \leq (\sigma^2(1,1)+\sigma^2(1,0))(1+10 C_N)  
\end{align}
Adding \eqref{eq:condmean} and \eqref{eq:condvar} gives
\begin{align*}
N \cdot \Var (\widehat{W}_{\mathcal{X}}) &= N \cdot  \Var (\E [ \widehat{W}_{\mathcal{X}} \mid \mathbf{X}, \mathbf{D}  ]) + \E [ N \cdot \Var (\widehat{W}_{\mathcal{X}} \mid \mathbf{X}, \mathbf{D}  ) ]\\
&< \sigma^2_{\mathcal{X}} + (\sigma^2(1,1)+\sigma^2(1,0))10 C_N +2 C_N.
\end{align*}

\noindent 
\textbf{Step 3 (MSE).}\;  Combining \eqref{eq:m0mean} and\eqref{eq:m0var} gives \eqref{eq:m0mse} since  $16 N^{-1} \leq C_N$ for all $N \geq 100$.
  \end{proof}

\begin{lemma}[MSE Ranking]
\label{lem:mylem21}
For any $\epsilon \in (0, N^{-1/2})$ and $N$ large enough, MSE ranking  \eqref{eq:mse}  holds.
\end{lemma}

\begin{proof}[Proof of Lemma \ref{lem:mylem21}]
Let $N \epsilon^2 \leq 1$ and $C_N = \sqrt{2.25 \ln N /N}$. Lemma \ref{lem:oracle} gives a lower bound on $MSE (\widehat W_{G^{*}})$
\begin{align*}
N \cdot MSE (\widehat W_{G^{*}}) >  (\sigma^2(1,0) + \sigma^2(0,1)) (1- 10C_N). 
\end{align*}
Lemma \ref{lem:m0mse} gives an upper bound on $MSE(\widehat{W}_{\mathcal{X}})$
\begin{align}
N \cdot MSE(\widehat{W}_{\mathcal{X}}) < 3/4+(\sigma^2(1,1)+\sigma^2(1,0))+[4+10(\sigma^2(1,1)+\sigma^2(1,0))]C_N
\end{align}
 Therefore, when $\sigma^2(0,1)-\sigma^2(1,1)-3/4>0$, there exists $N_0$ that depends on conditional variances such that 
 $$
N \cdot MSE (\widehat W_{G^{*}}) - N \cdot MSE(\widehat{W}_{\mathcal{X}})> 0.
$$
\end{proof}

\subsection{Proof of Proposition \ref{thm:equivalence}}
\label{sec:sec2}

We start with an auxiliary Lemma. Let $G^{*}  \triangle  G = G^{*} \setminus G \cup G \setminus G^{*}$ denote the symmetric difference of sets $G^{*}$ and $G$.  Let $ P(X \in G^{*}  \triangle  G)$ denote the share of people to be treated differently from the optimal policy, or the non-optimal share. Lemma A.7 in \citet{KitagawaTetenovSupp}, borrowing from \citet{Tsybakov}, bounds the welfare gap in terms of non-optimal share.  Lemma \ref{lem:tsybakov} complements this result by adding an upper bound on the standard deviation gap. 

\begin{lemma}
\label{lem:tsybakov}
 Suppose Assumptions \ref{ass:overlap} and \ref{ass:margin} hold.  Then, 
 (1) The welfare gap is bounded 
\begin{align}
C_B(P (X \in G^{*}  \triangle  G))^{1+1/\delta} &\leq W_{G^{*}} -W_G \leq M P (X \in G^{*}  \triangle  G), \label{eq:tsybakov}
\end{align}
where $C_B=\eta \delta(\frac{1}{1+\delta})^{1+1/\delta}>0$; \newline
(2) The variance gap is bounded as
\begin{align}
\label{eq:varineq}
\sigma^2_{G^{*}} - \sigma^2_G \leq \frac{5}{4}\frac{M^2}{\kappa} P (X \in G^{*}  \triangle  G).
\end{align}
(3) The standard deviation gap is bounded as
\begin{align}
\label{eq:varineq2}
\sigma_{G^{*}} - \sigma_G \leq \frac{5}{4} \frac{M^2}{2\bar{\sigma} \kappa} P (X \in G^{*}  \triangle  G).
\end{align}
\end{lemma}

\begin{proof}
\textbf{Step 1.} The lower bound \eqref{eq:tsybakov} is stated as Lemma A.7 in \citet{KitagawaTetenovSupp} and originally established in \citet{Tsybakov}. The upper bound is straightforward.

\textbf{Step 2.}  We introduce extra notation to simplify variance expressions. Given a policy $G$, let $G_1 = G$ and $G_0 = G^c$.  Then, the welfare $W_G$ in \eqref{eq:wfbstar} can be equivalently rewritten as  
 $$
 \textstyle  W_G = \E \left[\sum_{d \in \{1,0\}} m(d, X) \bm{1}\{ X \in G_d \}\right].
 $$
 Since $\bm{1}\{ X \in G  \} \bm{1} \{ X \in G^c \} =0 \text{ a.s.}$, we have
 $$
  \textstyle \E \left[ ( \sum_{d \in \{1,0\}} m(d, X) \bm{1}\{ X \in G_d \} )^2\right]  = \E \left[\sum_{d \in \{1,0\}} m^2(d, X) \bm{1}\{ X \in G_d \}\right].
 $$ 
Thus, by the Law of Total Variance,
\begin{align}
\textstyle \sigma^2_G = \sum_{d \in \{1,0\}} \E \bigg[ \bigg( \dfrac{ \sigma^2(d,X)}{ P (D=d \mid X) } + m^2(d,X) \bigg) \bm{1}\{ X \in G_d \} \bigg]- W^2_G. \label{eq:efficiencyboundg}
\end{align}
Denoting
\begin{align*} 
T_{1G}&= \E \bigg[ \bigg(\dfrac{\sigma^2(1,X)}{\pi(X)} -\dfrac{\sigma^2(0,X)}{1-\pi(X)} \bigg)   \bigg(\bm{1}\{ X \in G^{*} \setminus  G \} - \bm{1}\{ X \in  G \setminus G^{*}   \} \bigg) \bigg];  \\[3mm]
T_{2G}&= \E  \bigg[ \bigg(m^2(1,X) - m^2(0,X) \bigg) \bigg(\bm{1}\{ X \in G^{*} \setminus  G \} - \bm{1}\{ X \in  G \setminus G^{*}   \} \bigg) \bigg];\\[3mm]
T_{3G}&= - (W^2_{G^{*}}  - W^2_G),
\end{align*}
we can write
\begin{align}
\label{eq:sigmadiff}
\sigma^2_{G^{*}} - \sigma^2_G = T_{1G} + T_{2G} + T_{3G}.
\end{align}

\textbf{Step 3.}  By Assumption \ref{ass:overlap} and Jensen inequality, $\sigma^2(d, x) \leq m^2 (d,x) \leq M^2/4$, for all $d \in \{0, 1\}$, $x \in \mathcal{X}$. Thus, the first two terms are bounded as
\begin{align*}
T_{1G} &\leq  M^2/(2\kappa) \cdot P (X \in G^{*}  \triangle  G)\\
 T_{2G}&\leq M^2 / 2 \cdot P(X \in G^{*}  \triangle  G).
\end{align*}
The final term is bounded as
\begin{align*}
T_{3G} &\leq  |(W_{G^{*}}  - W_G)| |W_{G^{*}} + W_G|  \leq  M^2P (X \in G^{*}  \triangle  G).
\end{align*}
where the second inequality follows from $|W_{G^{*}} + W_G| \leq M$  and  \eqref{eq:tsybakov}. Collecting the terms and using the fact that $\kappa < 1/2$,
\[
T_{1G} + T_{2G} + T_{3G} \leq \frac{M^2}{\kappa} \frac{1 + 3\kappa}{2} \leq \frac{5}{4} \frac{M^2}{\kappa}, 
\]
so \eqref{eq:varineq} follows. \medskip 

\textbf{Step 4.} The corresponding bound on the standard derivation gap is
$$
\sigma_{G^{*}} - \sigma_G \overset{(i)}{\leq} \dfrac{ \sigma^2_{G^{*}} - \sigma^2_G}{2 \underline{\sigma}} \overset{(ii)}{\leq} \frac{5}{4}\frac{M^2}{2\underline{\sigma}\kappa} P (X \in G^{*}  \triangle  G)
$$
where (i) follows from $\sigma^2_{G^{*}} \geq \underline{\sigma}^2$ and $\sigma^2_{\mathcal{X}} \geq \underline{\sigma}^2$ (from Lemma \ref{lem:dr-cov-lower}) and (ii) from \eqref{eq:varineq}.
\end{proof}

\subsubsection{Proof of Proposition \ref{thm:equivalence}: Upper Bound}

Let $C_A= N^{-1/2} 1.25 z_{1-\alpha}M^2/(2 \underline{\sigma} \kappa)$ and $C_B$ be as defined in Lemma \ref{lem:tsybakov}. Define a function $g: \mathbf{R} \to \mathbf{R}$ as
\begin{equation}
\label{eq:maing}
g(x)=C_A x - C_B x^{1/\delta+1}.
\end{equation}
The LCB gap $\Delta_G$  is bounded as
$$
\Delta_G \overset{(i)}{\leq} C_A P (X \in G^{*}  \triangle  G)  - (W_{G^{*}} - W_G)  \overset{(ii)}{\leq}   g(P (X \in G^{*}  \triangle  G)),
$$
where (i) follows from \eqref{eq:varineq2} and (ii) from the lower bound in \eqref{eq:tsybakov}. Note that the function $g(x)$ is globally concave. Its' global maximum and maximizer are
\begin{align*}
g(x^{*}) =\left(\dfrac{C_A \delta }{C_B (1+\delta)}\right)^{\delta} \dfrac{C_A}{1+\delta}, \quad x^{*} = \left(\dfrac{C_A \delta }{C_B (1+\delta)}\right)^{\delta}.
\end{align*}
Therefore, for any $G \subseteq \mathcal{X}$,
\[
\Delta_{G} \leq g(x^*) = \overline{C} N^{-\frac{1 + \delta}{2}},
\]
where
 \[
  \overline{C}_{\delta, \eta} = \left( \frac{1.25 z_{1-\alpha} M^2}{2 \underline{\sigma} \kappa}\right)^{1 + \delta}  \frac{1}{\eta^{\delta}}.
 \]
 Under our assumptions, $\overline{C} = \max_{\delta, \eta}C_{\delta, \eta} < \infty$ and the slowest rate is attained at $\underline{\delta}.$

\subsubsection{Proof of Proposition \ref{thm:equivalence}: Lower Bound}

The proof is constructive and consists of three steps. Step 1 describes a class of DGPs. Step 2 shows that the proposed DGPs belong to the model $\mathbf{P}$. Step 3 establishes the lower bound.

\textbf{Step 1.}  Let $X \sim U[0,1]$, and the propensity score be constant, $\pi(x) = 1/2$, $\mathcal{X}$-a.s. Let  $\epsilon \in (0, 1/2)$ and $\nu > 0$ be a rational number for which the function $a \mapsto a^{\nu}$ is well-defined for both positive and negative values of $a$.\footnote{This is the case if and only if $\nu = \frac{p}{q}$, where $p, q$ are natural numbers with $\text{gcd}(p, q) = 1$.\label{foot:nu_rest}} Let $Y$ be a random variable supported on $[-M/2, M/2]$ so that the conditional means and second moments are bounded as $|m(d, x)| \leq M/2$ and $m^2(d, x)\leq M^2/4$. Consider the following specification:
\begin{equation} \label{eq:prop2_dgp}
\begin{array}{ll}
		m(1,x) = 0; & \sigma^2(1,x) = M^2/10; \\
		m(0,x) = - (x-\epsilon)^{\nu} M/5; &\sigma^2(0,x) = M^2/5.
\end{array}
\end{equation}
The CATE function is given by
$$
\tau(x) = m(1,x) - m(0,x) = (x-\epsilon)^{\nu} M / 5,
$$
so the first-best policy is
$$
G^{*} = [\epsilon,\, 1].
$$
Evidently, this distribution satisfies Assumption \ref{ass:overlap}. \medskip

\textbf{Step 2.}   We show that the proposed sequence of DGPs satisfies Assumption \ref{ass:margin} for a suitable choice of $\nu$. Note that for $t$ such that $(5t/M)^{1/\nu} \leq \epsilon$, 
$$
P (|X-\epsilon|^{\nu}M/5 \leq t) = P (|X-\epsilon| \leq (5t/M)^{1/\nu}) = (\epsilon + (5t/M)^{1/\nu}) - (\epsilon - (5t/M)^{1/\nu}) = 2(5t/M)^{1/\nu}.
$$
For $t$ such that $(5t/M)^{1/\nu} \geq \epsilon$,
$$
P (|X-\epsilon| \leq (5t/M)^{1/\nu}) \leq \epsilon + (5t/M)^{1/\nu} \leq 2(5t/M)^{1/\nu}.
$$
Thus, choosing $\nu$ such that $\delta = 1/\nu \geq \underline{\delta}$ but arbitrarily close to it,\footnote{Since rationals are dense in reals, it is without loss of generality to assume $\underline{\delta}$ is rational. Then, it can be expressed either as $\frac{p}{2^d q}$ where $p, d, q$ are natural numbers and $\text{gcd}(p, q) = 1$, or as $\frac{2^dp}{q}$ with the same conditions. In the former case, setting $\delta = 1/\nu = \underline{\delta}$ leads to $\nu$ satisfying the requirement of footnote \ref{foot:nu_rest}. In the latter case, setting $\delta = 1/\nu = \left( \frac{k}{k-1}\right)^d \frac{2^d p }{q}$ for any prime number $k$ corresponds to $\nu = \left(\frac{k-1}{2}\right)^2 \frac{q}{2^d p}$, which is also satisfies the requirement of footnote \ref{foot:nu_rest}. For arbitrarily large $k$, $1/\nu$ will be arbitrarily close to $\underline{\delta}$. } 

\[
P(|X - \epsilon|^{\nu}M/5 \leq t) \leq 2(5t/M)^{1/\nu} =  \left(\frac{t}{\eta}\right)^{\delta},
\]
so that \eqref{eq:margin} holds for any $\varepsilon \in (0, 1/2)$.  \medskip

\textbf{Step 3.} The first-best policy differs from $\mathcal{X}$ only for $x \in [0, \epsilon]$. Thus, the welfare gap is 
$$
W_{G^{*}} -W_{\mathcal{X}} = -\int_0^{\epsilon} (x-\epsilon)^{\nu}M/5 dx  = \frac{\epsilon^{\nu+1}}{\nu+1}M/5.
$$ 
The variance gap is obtained by plugging $G=\mathcal{X}$ into \eqref{eq:sigmadiff}. We have
\[
	T_{1\mathcal{X}} =  \frac{ \epsilon M^2}{5}; \;\;\;\;\;\;\; 	T_{2 \mathcal{X}} = \dfrac{\epsilon^{2\nu+1}}{2\nu+1} \frac{M^2}{25}; \;\;\;\;\;\;\; T_{3 \mathcal{X}} = - \dfrac{\epsilon^{\nu + 1}}{\nu + 1} \left( 2\dfrac{(1-\epsilon)^{\nu + 1}}{\nu + 1} - \dfrac{\epsilon^{\nu + 1}}{\nu + 1} \right)\frac{M^2}{25},
\]
so that
$$
\sigma^2_{G^{*}} -\sigma^2_{\mathcal{X}} = \frac{4M^2}{25}\epsilon + \underbrace{\left(\frac{\epsilon^{2\nu + 1}}{2\nu + 1} + \frac{\epsilon^{2\nu + 2}}{(\nu + 1)^2}\right) \frac{M^2}{25}}_{\geq 0}  + \underbrace{\left(\epsilon - \frac{2}{(\nu+1)^2} \epsilon^{\nu + 1}(1-\epsilon)^{\nu + 1}\right)\frac{M^2}{25}}_{=f(\epsilon)} > \frac{4M^2}{25}\epsilon,
$$
where the final inequality follows from the fact that $f'(\epsilon) \geq 0$ so $f(\epsilon)\geq f(0) = 0 $. On the other hand, recalling the DGP in \eqref{eq:prop2_dgp}, we can bound $\sigma^2_{G^*} < M^2/5(1 + \epsilon) < 3 M^2/10$ and $\sigma^2_{\mathcal{X}} < M^2/5$. Thus, $\sigma_{G^*} + \sigma_{\mathcal{X}} < M$, and
$$
\sigma_{G^{*}} -\sigma_{\mathcal{X}} = \frac{\sigma^2_{G^*} - \sigma^2_{\mathcal{X}}}{\sigma_{G^*} + \sigma_{\mathcal{X}}} >  \frac{4M}{25}\varepsilon.
$$
Setting $\epsilon^{\nu} = (4z_{1-\alpha}/5) N^{-1/2}$ and recalling that $\nu = 1/\delta$ gives a lower bound 
\[
\Delta_{\mathcal{X}}  > \underline{C} N^{-\frac{1 + \delta}{2}},
\]
where $\underline{C}_{\delta} = \frac{4 M}{5} \left( \frac{4z_{1-\alpha}}{5}\right)^{1+\delta} \frac{1}{1 + \delta}$. Since the above inequality holds for all for all $\delta$ arbitrarily close to $\underline{\delta}$, the stated result follows. 
 \qed

\section{Inverting Moment Inequality Tests}
\label{app:shb}

\subsection{Proof of Proposition \ref{prop:max_gms} }

The following lemma gives a closed-form solution for the lower confidence band based on inverting the Generalized Moment Selection test of \citet{andrews2010inference}.  Since the critical value of the GMS test is a step function, and the test statistic is a maximum of a finite number of linear functions, the confidence region obtained by test inversion may not be convex (although it can be shown that the probability of such an event approaches zero as $N$ increases). So, in the statement below, we conservatively define ${LCB}_{\mathcal{G}}^{GMS}$ as the lowest point of the confidence set obtained by test inversion.

\begin{lemma}[LCB based on GMS test inversion]
Denote: 
 \[
 \theta^{(j)} = \max \limits_{G \in \mathcal{G}}\left(\widehat{W}_G - \widehat{c}^{(j)}_{\alpha} \frac{\widehat{\sigma}_G}{\sqrt{N}} \right).
 \]
The lower confidence band obtained by inverting the GMS test takes the form:
\begin{align}
\label{eq:gms}
\widehat{LCB}_{\max}^{\text{GMS}} = \min\{\theta^{(j)}: t^{(j)} \geqslant \theta^{(j)} > t^{(j+1)}\}.
\end{align}
\end{lemma}
\begin{proof}
Under the GMS procedure, by definition, the critical value $\widehat{c}_{\alpha}(\theta)$ takes the form of a step function:
\[
\widehat{c}_{\alpha}(\theta) =  \sum_{j = 1}^{\mathcal{G}} \widehat{c}_{\alpha}^{(j)} \bm{1}(t^{(j)} \geqslant \theta > t^{(j+1)}).
\]
The function $T_N(\theta)$ is a maximum of a finite number of linear functions of $\theta$. The LCB corresponds to the lowest point of intersection between $T_N(\theta)$ and $\widehat{c}_N(\theta)$ (since the latter is a step function, there can be multiple such points). Each point $\theta^{(j)}$ marks the intersection of $T_N(\theta)$ with a constant function $\widehat{c}^{(j)}_{\alpha}$. If such $\theta^{(j)}$ is within the relevant ``step'' $[t^{(j)}, t^{(j+1)})$ of the critical value $\widehat{c}_{\alpha}(\theta)$, it is one of the intersection points of $T_N(\theta)$ and $\widehat{c}_{\alpha}(\theta)$. The minimum in the expression for $\widehat{LCB}_{\max}^{GMS}$ selects the lower point of intresection. 
\end{proof}

\subsection{Proof of Proposition \ref{prop:qlr}}
To simplify notation, we write $X - \theta \bm{1}$ instead of $\sqrt{N}(\widehat{W}_{\text{test}} - \theta \bm{1})$. For the strictly convex minimization problem:
	\[
	f^* = \min \limits_{t \in \mathbb{R}^d,\, t \leq 0} \{ (X - \theta \bm{1} - t)'\Sigma^{-1}(X - \theta \bm{1} - t) \},
	\]
	consider the dual objective function:
	\[
	g(u) = \min\limits_{t \in \mathbb{R}^d}\;\{ (X - \theta \bm{1} - t)'\Sigma^{-1}(X - \theta \bm{1} - t) + u't\},
	\]
	where \(u \geqslant 0\) is a vector of the Lagrange multipliers. Since the Slater condition holds, strong duality applies, so \(f^* = \max_{u \geqslant 0} g(u)\). Simple algebra yields
	\[
	g(u) = (X - \theta \bm{1})'u - \frac{1}{4}u'\Sigma u,
	\]
	so the event of not rejecting the LR test can be equivalently written as:
	\[
	\max \limits_{u \geqslant 0} \; \left\{ (X - \theta \bm{1})'u - \frac{1}{4}u'\Sigma u \right\} \leq c_{\alpha, LR}^{LF}.
	\]	
	For \(u = 0\), the inequality trivially holds, and for all \(u \geqslant 0\) with \(u \ne 0\) it is equivalent to
	\[
	\theta  \geqslant \frac{1}{(\sum_{j = 1}^d u_j)} \{ X'u - \frac{1}{4}u'\Sigma u - c_{\alpha, LR}^{LF} \}.
 	\]
 	Any $u \geqslant 0$ with $u \ne 0$ can be written as $u = \lambda \cdot \gamma$, where $\lambda \geqslant 0$ satisfies $\sum_{j = 1}^d \lambda_j = 1$, and $\gamma > 0$. Thus, the above display is equivalent to
	\[
	\theta \geqslant X'\lambda - \frac{1}{4}\lambda'\Sigma \lambda \cdot \gamma - \frac{c_{\alpha, LR}^{LF}}{\gamma}.
	\]
	Since this inequality holds for all \(\lambda \geqslant 0\) with \(\lambda'\bm{1} = 1\), and all \(\gamma > 0\), 
	\[
	\theta \geqslant \max \limits_{\lambda \in \Lambda, \gamma > 0} \left\{ X'\lambda - \frac{1}{4}\lambda'\Sigma \lambda \cdot \gamma - \frac{c_{\alpha, LR}^{LF}}{\gamma}\right\}.
	\]
	Concentrating out \(\gamma\) yields the stated result. \qed

\section{Auxiliary Empirical Details}
\label{app:emp}


\paragraph{Table \ref{tab:jtpa}, Row 4.}  To consider a data-driven choice of $G$,  we partition the sample into two parts $I_1$ and $I_2$ of sizes $N/3$ and $2/3 N$, respectively.  Let 
$$
\widehat G_1:= \{ X: \widehat \tau_1(X) >0 \},
$$
where $\widehat \tau_1$ is estimated via plug-in rule using random forest regression of earnings of $Educ$ and $PreEarn$.   A sample analog of $W_{gain,G}$ is 
$$
\widehat{W}_{gain,G}  = \dfrac{1}{| I_2| } \sum_{i \in I_2} \bigg( \dfrac{D_i}{\pi(X_i)} - \dfrac{1-D_i}{1-\pi(X_i)} \bigg) Y_i 1\{X_i \in G \}.
$$
Conditional on the data in the partition $I_1$, we have
$$
\sqrt{  | I_2| } ( \widehat{W}_{gain,\widehat G_1}  - W_{gain,\widehat G_1})  \Rightarrow^d N(0, \sigma^2_{\widehat G_1}) \mid  (W_i)_{i \in I_1}.
$$
The $100(1-\alpha)\%$ Lower Confidence Band defined as 
\begin{equation*}
{LCB}_{gain,G_1} = \widehat {W}_{gain,\widehat G_1}  - | I_2|^{-1/2} z_{1-\alpha} \widehat{\sigma}_{gain,\widehat G_1}
\end{equation*}
attains correct coverage condition on the data in $I_1$, and, therefore, unconditionally.

\end{document}